\documentclass[twoside]{article}

\usepackage[accepted]{aistats2026}
\usepackage[round]{natbib}

\newcommand{\thalf}{{\lceil T/2 \rceil}}
\usepackage{amsthm,amsmath,amsfonts,amssymb,stmaryrd}
\usepackage{subcaption}
\usepackage{graphicx}
\usepackage{float}
\usepackage[ruled, vlined]{algorithm2e}
\usepackage{xcolor}
\usepackage{hyperref}
\hypersetup{
    colorlinks,
    linkcolor={red!50!black},
    citecolor={blue!50!black}
}

\makeatletter
\def\namedlabel#1#2{\begingroup
    #2%
    \def\@currentlabel{#2}%
    \phantomsection\label{#1}\endgroup
}
\makeatother

\newtheorem{corollary}{Corollary}
\newtheorem{lemma}{Lemma}
\newtheorem{proposition}{Proposition}

\newtheorem*{assumptions*}{Assumptions}
\newtheorem{definition}{Definition}
\newtheorem*{example}{Example}

\DeclareMathOperator*{\argmin}{arg\,min}

\DeclareMathOperator*{\interior}{int}
\DeclareMathOperator*{\logdet}{log\,det}
\DeclareMathOperator*{\domain}{dom}
\DeclareMathOperator*{\tr}{tr}

\DeclareMathOperator*{\diag}{diag}

\newcommand{\proj}{\mathrm{proj}}

\begin{document}

%

%

\twocolumn[

\aistatstitle{Convergence of projected stochastic natural gradient variational inference for various step size and sample or batch size schedules}

\aistatsauthor{ Thomas Guilmeau \And Hadrien Hendrikx \And  Florence Forbes }

\aistatsaddress{ Univ. Grenoble Alpes, Inria, CNRS, Grenoble INP, LJK, 38000 Grenoble, France  } ]

\begin{abstract}
Stochastic natural gradient variational inference (NGVI) is a popular and efficient algorithm for Bayesian inference. Despite empirical success, the convergence of this method is still not fully understood. In this work, we define and study a projected stochastic NGVI when variational distributions form an exponential family. Stochasticity arises when either gradients are intractable expectations or large sums. We prove new non-asymptotic convergence results for combinations of constant or decreasing step sizes and constant or increasing sample/batch sizes. When all hyperparameters are fixed, NGVI is shown to converge geometrically to a neighborhood of the optimum, while we establish convergence to the optimum with rates of the form $\mathcal{O}\left(\frac{1}{T^{\rho}} \right)$, possibly with $\rho \geq 1$, for all other combinations of step size and sample/batch size schedules. These rates apply when the target posterior distribution is close in some sense to the considered exponential family. Our theoretical results extend existing NGVI and stochastic optimization results and provide more flexibility to adjust, in a principled way, step sizes and sample/batch sizes in order to meet speed, resources, or accuracy constraints.

\end{abstract}

\section{INTRODUCTION}
\label{sec:intro}

In Bayesian statistics, posterior distributions are often  intractable and require sophisticated inference techniques. Variational inference (VI) methods, which approximate posterior distributions
by optimizing over a family of tractable densities the Kullback–Leibler divergence from the true posterior, have emerged as an efficient alternative to Markov chain Monte Carlo \citep{blei2017, zhang2019}. VI algorithms can be easily deployed in many settings \citep{ranganath2014}, and the  properties of the optimal solution are increasingly studied from many perspectives \citep{wang2018, alquier2020, margossian2025}. 

In order to solve the VI optimization problem, natural gradient descent, initially proposed by \cite{amari1998}, has been used successfully in many settings such as latent Dirichlet allocation with very large data-sets \citep{hoffman2013}, Bayesian neural networks \citep{khan2018}, or problems with discrete latent variables \citep{ji2021}. Natural gradient descent has also been used in black-box optimization \citep{ollivier2017}, machine learning \citep{martens2020}, or in the context of the so-called Bayesian learning rule \citep{khan2023}. Natural gradient descent preconditions the standard gradient by the inverse Fisher information matrix of the approximating distribution. The intuition behind its good performance is that it approaches Newton's method \citep{martens2020}. In the case of exponential families \citep{barndorff-nielsen2014}, natural gradient descent is equivalent to mirror descent \citep{raskutti2015, wu2024}, whose ability to better tailor the geometry of the objective function has been shown by e.g., \cite{bauschke2016, lu2018}.

Despite these empirical good performance and intuitions behind its working, the convergence of stochastic NGVI is still poorly understood. In the deterministic setting, let us mention the works of \cite{kumar2025, godichon-baggioni2025}. In the stochastic setting, \cite{hoffman2013} claim that convergence can be established through results by \cite{bottou1999}, which may not apply without further assumptions. \cite{khan2016} use the equivalence between NGVI and mirror descent when exponential families are considered  to derive a convergence rate to approximate stationarity under a strong convexity assumption which may fail in some parts of the search space, even for Gaussian distributions \citep{guilmeau2025}. Still in the context of exponential families, \cite{wu2024} leverage 
a notion of variance from \cite{hanzely2021} to prove a $\mathcal{O}\left( \frac{1}{T}\right)$ convergence rate when the posterior distribution belongs to the considered exponential family, and they formulate this variance-like quantity only in a conjugate Bayesian linear regression setting with stochasticity coming from subsampling the data. The same notion of variance is also used by \cite{sun2025} to derive convergence rates in possibly non-convex settings but restricted to mean-field Gaussian approximating families.

In this work, we also use the equivalence between natural gradient descent and mirror descent. We use and refine novel analysis techniques for stochastic mirror descent from \cite{hendrikx2024} to improve the generality and applicability of existing results for stochastic NGVI. In this endeavor, we make the following contributions. $\mathbf{(1)}$ We propose a projection step in a geometry that is compatible with the mirror step with novel non-asymptotic convergence rates for the resulting projected stochastic mirror descent algorithm, which may be of independent interest. We establish geometric convergence to a neighborhood of the minimizer for constant step sizes and gradient estimators computed with fixed sample/batch sizes. We establish $\mathcal{O}\left(\frac{1}{T^{\rho}}\right)$ convergence to the minimizer for either fixed step sizes and increasing sample/batch sizes or diminishing step sizes and fixed or increasing sample/batch sizes, with possibly $\rho \geq 1$. $\mathbf{(2)}$ We exhibit a new condition on the posterior distribution, which ensures that our general convergence rates apply to our projected stochastic NGVI algorithm, while allowing posterior distributions that do not belong to the considered exponential family. $\mathbf{(3)}$ Our results involve a variance-like term proposed by \cite{hendrikx2024}, which we upper-bound in case of estimators based on samples from the approximating distribution or on sub-sampled data batches, thus providing a precise dependence on the step size and sample/batch size in our results. It follows an increased  number of possible strategies to jointly schedule step and sample/batch sizes, depending on the desired convergence speed, accuracy, and data-efficiency.

After specifying some  background  and  our stochastic NGVI algorithm in Sections \ref{sec:background} 
\& \ref{sec:presentation}, our non-asymptotic convergence rates for different choices of step sizes and sample/batch sizes are presented in Section \ref{sec:convergence}, with practical and numerical illustrations in  Sections \ref{sec:Illus} \& \ref{sec:experiments}. Related work is discussed in Section \ref{sec:related}. Proofs and additional results are postponed to the Appendix.

\section{BACKGROUND}
\label{sec:background}

\paragraph{Notation} For a measurable space $\mathcal{X}$ and a measure $\nu$ on $\mathcal{X}$, $\mathcal{P}(\mathcal{X},\nu)$ denotes the set of probability densities $p$  with respect to $\nu$ on $\mathcal{X}$ and $\mathbb{E}_p$ the expectation with respect to $p$. $\mathcal{H}$ 
denotes a finite-dimensional Hilbert space with scalar product  denoted by $\langle \cdot, \cdot \rangle$. $\mathbb{S}^d$ (resp.~$\mathbb{S}^d_{> 0}$)  denotes the space of symmetric (resp.~positive definite) $d \times d$ matrices. 
$\mathcal{N}(\mu, \Sigma)$ is the Gaussian distribution with mean $\mu$ and covariance matrix $\Sigma$ while $\mathcal{U}_M$ is the uniform distribution on $\llbracket 1,M \rrbracket$. For a proper function $h:\mathcal{H}\rightarrow\mathbb{R} \cup \{\infty\}$ with domain $\domain h$ which is differentiable on $\interior \domain h$ with gradient $\nabla h$, $d_h$ denotes the following quantity
 $$d_h(\omega, \omega') = h(\omega) - h(\omega') - \langle \nabla h(\omega'), \omega - \omega'\rangle$$ for all $(\omega, \omega') \in {\cal H} \times \interior \domain h$. When $h$ is strictly convex then $d_h(\omega,\omega') \geq 0$ with equality if and only if $\omega = \omega'$ and $d_h$is the induced {\bf Bregman divergence}.

\paragraph{Variational inference}

Variational inference efficiently approximates intractable Bayesian posterior distributions resulting from updating prior belief $X \sim p_0$ upon observing $Y$. We denote by $\pi \in \mathcal{P}(\mathcal{X},\nu)$ the density of $X | Y$, given by Bayes rule, 
$ \pi(x) \propto p(y | x) p_0(x)$. VI aims at approximating $\pi$ by a density from a family $\mathcal{Q} \subset \mathcal{P}(\mathcal{X}, \nu)$. This is done by solving
\begin{equation}
    \label{eq:VIproblem}
    \min_{q \in \mathcal{Q}} \text{KL}(q, \pi),
\end{equation}
where $\text{KL}(p_1, p_2) = \mathbb{E}_{p_1}[\log p_1(X)] - \mathbb{E}_{p_1}[\log p_2(X)]$, for $p_1,p_2 \in \mathcal{P}(\mathcal{X},\nu)$,   is the Kullback-Leibler divergence.

When $\pi \in \mathcal{Q}$, it is possible to exactly recover $\pi$. However, in most settings, $\pi \notin \mathcal{Q}$.

\paragraph{Exponential family}

We will consider $\mathcal{Q}$ in problem \eqref{eq:VIproblem} to be an exponential family \citep{brown1986, barndorff-nielsen2014}.

\begin{definition}
The exponential family with sufficient statistic $\Gamma : \mathcal{X} \rightarrow \mathcal{H}$ is the set $\mathcal{Q} \subset \mathcal{P}(\mathcal{X}, \nu)$ s.t.~every $q \in \mathcal{Q}$ satisfies $q = q_\theta$ for some $\theta \in \domain A$ with
    \begin{equation}
        \label{eq:expFamilyDensity}
        q_\theta(x) = \exp( \langle \theta, \Gamma (x) \rangle - A(\theta)),\,\nu\text{-a.e.},
    \end{equation}
    where $\theta$ is the {\it natural parameter} and $A$ is the log-partition function  $A(\theta) = \log \left( \int \exp(\langle \theta, \Gamma(x) \rangle) \nu(dx) \right)$. 
\end{definition}
For $q_\theta \in \mathcal{Q}$, the associated {\bf expectation parameter} $\omega$ is defined as the expected sufficient statistic
\begin{equation*}
    \omega =\mathbb{E}_{q_{\theta}}[\Gamma(X)].
\end{equation*}
For our proofs, we will need to assume that $\mathcal{Q}$ is steep. Steepness is defined in \cite[Chapter 8]{barndorff-nielsen2014} and recalled in  Appendix \ref{app:expFam}. It is implied by $\domain A$ being open. For example, Gaussian distributions form a steep exponential family.

The log-partition function benefits from convex analysis properties. Introducing its {\bf convex conjugate}
$$A^*(\omega) = \sup_{\theta \in \mathcal{H}} \left\{ \langle \theta, \omega \rangle - A(\theta) \right\},$$ the following properties hold.
\begin{proposition}
    \label{prop:expFamilyProperties}
    Suppose that $\interior \domain A \neq \emptyset$. Then,
    \begin{itemize}
        \item [$(i)$] $A$ is proper, lower-semicontinuous, strictly convex, and 
        differentiable
        on $\interior \domain A$. The gradient of $A$ links the expectation and natural parameters with $\omega =\mathbb{E}_{q_{\theta}}[\Gamma(X)] = \nabla A(\theta).$
        
        The Fisher information matrix of $q_{\theta}$ takes the form 
        $\mathbb{E}_{q_{\theta}}[\nabla \log q_\theta(X)(\nabla \log q_\theta(X))^\top] = \nabla^2 A(\theta).$
    \end{itemize}
    Furthermore, if $\mathcal{Q}$ is minimal and steep, then
    \begin{itemize}
        \item [$(ii)$] $A^*$ is proper, lower semi-continuous, strictly convex, differentiable on $\interior \domain A^*\!$. If $\omega \!\in \!\interior \domain A^*$ with $\omega \!= \!\nabla A(\theta)$, then $A^*(\omega) \!=\! \mathbb{E}_{q_{\theta}}[\log(q_{\theta}(X))]$.
        \item [$(iii)$] $\nabla A$ is a bijection from $\interior \domain A$ to $\interior \domain A^*$, its inverse is $\nabla A^*$ and 
        $$\omega  = \nabla A(\theta) \iff \theta = \nabla A^*(\omega).$$
        \item [$(iv)$] For every $\theta, \theta' \in \interior \domain A$ and their respective expectation parameters $\omega = \nabla A(\theta), \omega' = \nabla A(\theta')$, 
        $$d_A(\theta, \theta') = d_{A^*}(\omega', \omega) = \text{KL}(q_{\theta'}, q_{\theta}).$$
    \end{itemize}
\end{proposition}

Proposition~\ref{prop:expFamilyProperties} implies that any density in the exponential family $\mathcal{Q}$ can be equivalently represented by $q_{\theta}$ or $q_{\omega}$, with $\omega = \nabla A(\theta)$ or equivalently $\theta = \nabla A^*(\omega)$, motivating the following definitions.

\begin{definition}
    \label{def:VIObjective}
    Consider an exponential family $\mathcal{Q}$. Then, we define the functions $f^P_{\pi}$ and $f^D_{\pi}$ as
    \begin{align*}
        f^P_\pi(\theta) &= \text{KL}(q_{\theta}, \pi),\,\forall \theta \in \domain A,\\
        f^D_\pi(\omega) &= \text{KL}(q_{\omega}, \pi),\,\forall \omega \in \domain A^*.
    \end{align*}
\end{definition}
Remark that $f^P_\pi \circ \nabla A^* = f^D_\pi$ while $f^D_\pi \circ \nabla A = f^P_\pi$.

\paragraph{Natural gradient variational inference}

Natural gradient descent is a preconditioned gradient descent where the conditioning is done by the Fisher information matrix of the current distribution  $q_\theta$ \citep{amari1998}. 
When $q_\theta$ is in an exponential family  ${\cal Q}$, 
 Proposition~\ref{prop:expFamilyProperties}~(i) states that the Fisher information matrix at $q_{\theta}$ is $\nabla^2 A(\theta)$. To
 solve \eqref{eq:VIproblem}, or from Definition \ref{def:VIObjective},  to equivalently minimize $f^P_\pi$ on $\domain A$ or  $f^D_\pi$ on $\domain A^*$,  
 a natural gradient descent update for $f^P_\pi$, with step size $\eta > 0$, updates $\theta \in \domain A$ to $\theta_+$ set to 
\begin{equation}
    \label{eq:NGVIupdate}
    \theta_+ = \theta - \eta \left( \nabla^2 A(\theta) \right)^{-1} \nabla f^P_\pi(\theta).
\end{equation}
In addition, $\nabla f^P_\pi(\theta) = \nabla^2 A(\theta) \nabla f^D_\pi(\nabla A(\theta))$ so that \eqref{eq:NGVIupdate} is equivalent to
\begin{equation}
    \label{eq:NGVIMirror}
    \theta_+ = \theta - \eta \nabla f^D_\pi(\omega), \quad \mbox{where $\omega = \nabla A(\theta)$.}
\end{equation}
The update \eqref{eq:NGVIMirror} is a mirror descent update with mirror map $\nabla A^*$: $\nabla A^*(\omega_+) = \nabla A^*(\omega) - \eta \nabla f_{\pi}^D(\omega)$  with $\omega_+ = \nabla A(\theta_+)$ (see also \cite[Theorem 1]{raskutti2015} and \cite{wu2024}). 

Although \eqref{eq:NGVIMirror} alleviates the need to invert a Hessian matrix at each iteration, we still need to compute gradients $\nabla f^D_\pi(\omega)$. From Proposition~\ref{prop:expFamilyProperties} (ii), it comes  
\begin{align}
f^D_\pi(\omega) & = A^*(\omega) - \mathbb{E}_{q_\omega}[\log \pi(X)], \label{def:lulu}
\end{align}
so that using (iii), (\ref{eq:NGVIMirror}) writes as 
\begin{equation}
    \label{eq:updateNGVIasMirror}
    \theta_+ = (1-\eta) \theta + \eta \nabla \mathbb{E}_{q_\omega}[\log \pi(X)].
\end{equation}
The real challenge is then to compute $\nabla \mathbb{E}_{q_\omega}[\log \pi(X)]$, which has been approached from several perspectives. There are two fundamental situations.

{\textbf{1)} Monte Carlo approximation:} The gradient $\nabla \mathbb{E}_{q_{\omega}} [ \log \pi(X)]$ is intractable but can be approximated by samples from $q_{\omega}$. Let us mention the gradient estimators based on Fisher's identity used by \cite{ranganath2014, ji2021}, gradient estimators based on reparametrization \citep{domke2023, kim2023}, or in the case of Gaussian distributions, gradient estimators based on \citep{bonnet1964, price1958} used  by \cite{rezende2014, khan2023}.

{\textbf{2)} Data subsampling:} The gradient is a finite sum of tractable terms but the sum is too large to be computed.
Data subsampling consists in approximating the gradient by only considering a randomly chosen subset of terms appearing in the sum, referred to as a batch. See for instance \citep{hoffman2013}.

\paragraph{Bregman projection and stochastic mirror descent}
\label{subsec:BregmanBackground}

We outlined above the link between natural gradient descent and mirror descent, which is itself related to the geometry induced by Bregman divergences \citep{bauschke2016, lu2018}. Motivated by these links, we recall some useful notions from the literature of optimization in Bregman geometries.

In this part, we  consider a Legendre function $h$. If $h$ has an open domain, is differentiable and is strictly convex, then it is Legendre, see Appendix~\ref{app:proofs} for more details. The Bregman projection below associated to $d_h$ will be useful to enforce constraints while respecting the geometry induced by natural gradient descent.
\begin{definition}
    \label{def:BregmanProj}
    For $C \subset {\cal H}$, the Bregman projection $\proj_C^h $ is defined for any $\omega \in \interior \domain h$ as
    \begin{equation*}
        \proj_C^h(\omega )= \argmin_{\omega' \in C} \left\{ d_h(\omega', \omega) \right\}.
    \end{equation*}
\end{definition}
\begin{proposition}
    \label{prop:BregmanProjection}
    If $C$ is closed and convex, such that $C \cap \interior \domain h \neq \emptyset$ with $h$ being Legendre, then $\proj_C^h$ is well-defined and single-valued on $\interior \domain h$, and $\proj_C^h(\omega) \in \interior \domain h$ for any $\omega \in \interior \domain h$.
\end{proposition}
We  review recent results about the convergence of stochastic mirror descent algorithms. Consider the optimization problem of finding $\omega_*$ s.t.~
\begin{equation*}
    f(\omega_*) = \min_{\omega \in D} f(\omega),
\end{equation*}
when $f$ is convex and its gradient is unavailable, e.g., because it involves an expectation as in  $f^D_\pi$. 
We assume $D \!\subset\! \domain h$.
The stochastic mirror descent algorithm with mirror map $\nabla h$ and step size $\eta > 0$, updates $\omega$ to $\omega_+$, assumed to be in $D$, as
\begin{equation}
    \label{eq:omegaPlus}
    \nabla h(\omega_+) =  \nabla h(\omega) - \eta G^N(\omega),
\end{equation}
where $G^N(\omega)$ is an unbiased estimator of $\nabla f (\omega)$ that depends on a parameter $N \in \mathbb{N}_{> 0}$ representing e.g., the sample size or the batch size used to compute it. The convergence of such an algorithm is difficult to establish because of the interaction between the noise and the non-Euclidean geometry of the algorithm. In particular, the effect of the noise depends on the location of the current iterate.

In order to account for this specific geometry, several Bregman-based generalizations of Euclidean notions, such as variance, strong convexity, or smoothness have been introduced.
A Bregman-based variance-like quantity is defined by \cite{hendrikx2024}. 
\begin{definition}
    \label{def:variance}
    Consider the step size $\eta > 0$ and $N \in \mathbb{N}_{>0}$. Then, we define the function $f_{\eta,N} : D \rightarrow \mathbb{R}$ as
    \begin{equation*}
        f_{\eta, N}(\omega) = f(\omega) - \frac{1}{\eta} \mathbb{E}\left[ d_h(\omega, \omega_+) \right],\,\forall \omega \in D,
    \end{equation*}     
    with $\omega_+$ defined as in \eqref{eq:omegaPlus}. The variance $\sigma_{\eta,N}^2$  is then
    \begin{equation*}
        \sigma_{\eta,N}^2 = \frac{1}{\eta} \sup_{\omega \in D} \left\{f(\omega_*) - f_{\eta,N}(\omega) \right\}.
    \end{equation*}
\end{definition}
It is smaller than previously proposed notions such as that of \cite{hanzely2021}, and \citet[Proposition 2.1]{hendrikx2024} states that $\sigma_{\eta,N}^2 \geq 0,\,\forall \eta > 0$.

We recall the notions of relative strong convexity and relative smoothness, which generalize the notions of strong convexity and smoothness to the Bregman geometry \citep{bauschke2016, lu2018}.

\begin{definition}
    \label{def:relativeProperties}
Let $m \in \mathbb{R}_{>0}$,   $f$ is said to be $m$-strongly-convex relatively to $h$ if the following holds: 
       $ m \, d_h(\omega, \omega') \leq d_f(\omega, \omega'),\,\forall \omega, \omega' \in D.$
    
    Similarly, for $\ell \in \mathbb{R}_{>0}$, $f$ is said to be $\ell$-smooth relatively to $h$ if 
       $d_f(\omega, \omega') \leq \ell \,d_h(\omega, \omega'),\,\forall \omega, \omega' \in D.$
\end{definition}

Relative smoothness allows to derive an upper-bound on $\sigma_{\eta,N}^2$, adapted from \citet[Prop.~2.4]{hendrikx2024}.
\begin{proposition}
    \label{prop:varianceBoundHadrien}
    For $\omega \in D$, we define $\nabla h(\overline{\omega}_+) = \nabla h(\omega) - \eta \nabla f(\omega)$, which is the exact counterpart of $\omega_+$ in (\ref{eq:omegaPlus}). If $f$ is $\ell$-smooth relatively to $h$ and $\eta \leq \frac{1}{\ell}$, then
    \begin{equation*}
        \sigma_{\eta,N}^2 \leq \frac{1}{\eta^2} \sup_{\omega \in D} \left\{ \mathbb{E}\left[ d_h(\overline{\omega}_+, \omega_+)\right] \right\}.
    \end{equation*}
\end{proposition}
Then, \citet[Theorem 3.1]{hendrikx2024} states that if $f$ is $m$-strongly-convex relatively to $h$, the sequence $\{\omega_t\!\}_{\!t\in\mathbb{N}}$ generated by iterating \eqref{eq:omegaPlus} satisfies at every $t \in \mathbb{N}$
\begin{equation*}
    \begin{split}
        \eta \!\left( \!\mathbb{E}\left[ f_{\eta,N}(\omega_t) \right] - \inf_{\omega \in D} f_{\eta,N}(\omega) \!\right) + \mathbb{E} \left[ d_h(\omega_*, \omega_{t+1}) \right]\\ \leq (1 - m \eta)^{t+1} d_h(\omega_*, \omega_0) + \frac{\eta \sigma_{\eta,N}^2}{m}.
    \end{split}
\end{equation*}
When $\eta \leq 1/m$, we thus have geometric convergence to a small neighbourhood of $\omega_*$, whose size is controlled by $\sigma_{\eta,N}^2$.

\section{PROJECTED STOCHASTIC NGVI}
\label{sec:presentation}

We consider the constrained VI problem
\begin{equation}
    \label{eq:VIExpFamConstraint}
    \min_{\omega \in C \cap \domain A^*} f^D_\pi(\omega),
\end{equation}
which is an instance of Problem~\eqref{eq:VIproblem} with approximating family $\!\{ q_\omega \!\in \!\mathcal{Q},\,\omega \!\in\! C \!\cap\! \domain A^*\}$, meaning that we restrict densities from $\mathcal{Q}$ to have parameters in $C$. For  Gaussian distributions, $C$ may for instance enforce covariance matrices with constrained eigenvalues.

\subsection{Projected stochastic natural gradient}

To solve \eqref{eq:VIExpFamConstraint},
recall from Equation~\eqref{eq:updateNGVIasMirror} that an idealised natural gradient update for minimizing $f^D_\pi$ reads
\begin{equation*}
    \nabla A^*(\omega_+) = (1-\eta) \nabla A^*(\omega) + \eta \nabla \mathbb{E}_{q_\omega}[\log \pi(X)],
\end{equation*}
where we need to approximate $\nabla \mathbb{E}_{q_\omega}[\log \pi(X)]$. We suppose that we have an estimator $g^N(\omega)$ of $\nabla \mathbb{E}_{q_\omega}[\log \pi(X)]$ where $N$ is a sample or batch size. Using formulation \eqref{eq:NGVIMirror}, this amounts to approximate $\nabla f^D_\pi(\omega)$ by $G^N(\omega) = \nabla A^*(\omega) - g^N(\omega)$.

We propose the following Algorithm~\ref{alg}, a projected mirror descent algorithm, with possibly varying step and sample/batch sizes, which alternates stochastic mirror descent steps as in \eqref{eq:omegaPlus} and Bregman projection on $C$ as defined in Definition~\ref{def:BregmanProj}. We detail specific instances of $g^N$ in Sections~\ref{subsec:BonnetPrice} and \ref{subsec:BayesianLinearRegression}. Examples of constraint sets $C$ which admit projection operators $\proj_C^{A^*}$ with closed form are provided in Appendix~\ref{app:projections}.
\begin{algorithm}
    \SetAlgoLined
    Select step sizes $\{\eta_t \}_{t \in \mathbb{N}}$ in $(0,1]$, $\{ N_t \}_{ t\in \mathbb{N}}$ in $\mathbb{N}_{>0}$, and a starting point $\omega_0 \in \interior \domain A^*$.\\
    \For{$t \geq 0$}{
\vspace{-.3cm}
    \begin{align*}
        \nabla A^* ((\omega_t)_+) &= (1 - \eta_t) \nabla A^*(\omega_t) + \eta_t g^{N_t}(\omega_t)\\
        \omega_{t+1} &= \proj_C^{A^*}((\omega_t)_+).
    \end{align*}
    \vspace{-.4cm}
    }
    \caption{Projected NGVI}
    \label{alg}
\end{algorithm}

\subsection{Assumptions}

We describe below our main assumptions, which will hold throughout the rest of this work.
    \begin{itemize}
        \vspace{-.3cm}
        \item[\namedlabel{assumption:steep}{$\mathrm{(A1)}$}] $\interior \domain A \neq \emptyset$ and ${\cal Q}$ is minimal and steep.
        \vspace{-.2cm}
        \item[\namedlabel{assumption:constraint}{$\mathrm{(A2)}$}] $C$ is closed, convex, and $C \cap \interior \domain A^* \neq \emptyset$.
        \vspace{-.2cm}
        \item[\namedlabel{assumption:wellposedness}{$\mathrm{(A3)}$}] Iterates $\{\omega_t\}_{t \in \mathbb{N}}$ from Algorithm \ref{alg} satisfy $\omega_t \in \interior \domain A^*$ almost surely for any $t \in \mathbb{N}$.
        \vspace{-.2cm}
        \item[\namedlabel{assumption:estimator}{$\mathrm{(A4)}$}] For any $\omega \in \interior \domain A^*$ and $N \in \mathbb{N}_{>0}$, $\mathbb{E}\left[ g^N(\omega) \right] = \nabla \mathbb{E}_{q_\omega}[\log \pi(X)]$.
        \vspace{-.2cm}
        \item[\namedlabel{assumption:interior}{$\mathrm{(A5)}$}] Problem \eqref{eq:VIExpFamConstraint} admits a unique solution $\omega_* \in C \cap \interior \domain A^*$.
    \end{itemize}

Assumptions \ref{assumption:steep} and \ref{assumption:constraint} are mild assumptions on $\mathcal{Q}$ and $C$ that allow us to benefit from Propositions \ref{prop:expFamilyProperties} and \ref{prop:BregmanProjection}. Assumption \ref{assumption:wellposedness} ensures that the iterates are well-posed. 
As shown in Appendix \ref{app:proof3}, \ref{assumption:wellposedness} is satisfied if $g^N(\omega) \in \interior \domain A$ almost surely for every $\omega \in \interior \domain A^*$, $N \in \mathbb{N}_{>0}$, or if the step sizes $\{ \eta_t \}_{t \in \mathbb{N}}$ are small enough. 
Assumption \ref{assumption:estimator} is a mild unbiasedness assumption on our estimators.
Finally, Assumption \ref{assumption:interior} requires the considered VI problem to have a unique solution that must belong to $\interior \domain A^*$, e.g., excluding cases where the optimum is reached by a singular covariance in the Gaussian case. A sufficient condition for \ref{assumption:interior} is provided in subsection \ref{subsec:sc}. 
We stress that $C \cap \interior \domain A^*$ will play the role of $D$ in Section~\ref{subsec:BregmanBackground}. In particular, the supremum in the definition of $\sigma_{\eta,N}^2$ will be taken over this set.

\section{CONVERGENCE OF PROJECTED NGVI}
\label{sec:convergence}

In this section, we provide two novel general convergence rates for Algorithm~\ref{alg} under different assumptions on the step sizes and sample/batch sizes. Both results require the relative strong convexity of $f^D_\pi$, as defined in  Definition~\ref{def:relativeProperties}, and involve the variance-like quantity $\sigma_{\eta, N}^2$, specified in  Definition~\ref{def:variance}. Although Algorithm~\ref{alg} refers to our projected stochastic NGVI algorithm, we stress that these results apply in the general context of projected stochastic mirror descent.

\paragraph{Fixed step sizes} Our first  Proposition \ref{prop:cvgceStochMD} assumes fixed $\eta_t\!\equiv \!\eta $. It states that Algorithm~\ref{alg} with $N_t \!\equiv\! N$ produces iterates that geometrically converge to a Bregman ball around $\omega_*$. If $N_t \!\!= \!\!(t+1)^\gamma$, our result implies that convergence to
$\omega_*$ is guaranteed and controlled by a sum of geometric terms and a $\mathcal{O}\!\left(\frac{1}{T^\gamma} \right)$ term.
\begin{proposition}
    \label{prop:cvgceStochMD}
    Assume $f^D_\pi$ is $m$-strongly-convex relatively to $A^*$.
    Let $\{ \omega_t\}_{t \in \mathbb{N}}$ be generated by Algorithm~\ref{alg} with $\eta_t \equiv \eta \in (0,m^{-1}]$. If $N_t \equiv N \in \mathbb{N}_{>0}$, we have for any $T \in \mathbb{N}$ that
    \begin{equation*}
        \mathbb{E} \left[ d_{A^*}(\omega_*, \omega_T) \right] \leq (1 - m \eta)^T d_{A^*}(\omega_*, \omega_0) + \frac{\eta \sigma_{\eta,N}^2}{m}.
    \end{equation*}
    If there exists $V > 0$ s.t.~$\sigma^2_{\eta, N} \leq \frac{V}{N}$ for any $\eta \in (0,m^{-1}]$, and $N_t = (t+1)^\gamma$, $\gamma > 0$, we have for any $T \in \mathbb{N}$ that
    \begin{multline*}
        \mathbb{E} \left[ d_{A^*}(\omega_*, \omega_T) \right] \leq (1 - m \eta)^T d_{A^*}(\omega_*, \omega_0)\\ + \frac{V}{m \eta}\left( (1-m\eta)^{\frac{T+1}{2}} + \frac{2^\gamma}{T^\gamma} \right).
    \end{multline*}
\end{proposition}

\paragraph{Decreasing step sizes} Our second result in Proposition \ref{prop:cvgceStochMDDecreasing} is set for decreasing step sizes. It gives convergence rates for the exact convergence of the iterates to $\omega_*$. If $N_t \equiv N$, this convergence is in $\mathcal{O}\left(\frac{1}{T} \right)$, but it can be made faster if the sample/batch size increases with iterations. For instance, if $N_t = (t+1)^\gamma$,  for some $\gamma \in (0,1)$, then the convergence is in $\mathcal{O}\left(\frac{1}{T^{1 + \gamma}}\right)$.
\begin{proposition}
    \label{prop:cvgceStochMDDecreasing}
    Assume $f^D_\pi$ is $m$-strongly-convex relatively to $A^*$ and  there exists $V > 0$ s.t.~$\sigma_{\eta,N}^2 \leq \frac{V}{N}$ for any $\eta \in (0,m^{-1}]$ and $N \in \mathbb{N}_{>0}$. 
    Consider $\{ \omega_t\}_{t \in \mathbb{N}}$ generated by Algorithm~\ref{alg} with $\eta_t = \frac{1}{m(t/2+1)}$. Then, we have for any $T \in \mathbb{N}$ that
    \begin{equation*}
        \mathbb{E} \left[ d_{A^*}(\omega_*, \omega_{T+1}) \right] \leq \frac{4V}{m^2 (T+2)(T+1)}\sum_{t=0}^T \frac{1}{N_t}. 
    \end{equation*}
\end{proposition}

\subsection{Sufficient conditions}

\label{subsec:sc}

We now introduce a new condition ensuring \ref{assumption:interior}, that $f^D_\pi$ is relatively strongly convex, and that allows to control $\sigma^2_{\eta, N}$. This condition requires $\pi$ to belong to an exponential family that extends $\mathcal{Q}$ in a precise way, without necessarily requiring $\pi \in \mathcal{Q}$ as often done in previous work, e.g., \citep{wu2024}.

\begin{definition}
    \label{def:linearEnlargement}
    The posterior $\pi$ satisfies the linearly extended recoverability condition (LERC) with respect to ${\cal Q}$ if $\pi$ belongs to an exponential family $\mathcal{Q}_\pi$ with sufficient statistic $\Gamma_\pi : \mathcal{X} \rightarrow \mathcal{H}_\pi$ for which there exists a linear operator $L_\pi : \mathcal{H} \rightarrow \mathcal{H}_\pi$ satisfying $\mathbb{E}_{q_\omega}[\Gamma_\pi(X)] = L_\pi \mathbb{E}_{q_\omega}[\Gamma(X)]$ for all $\omega \in \domain A^*$.
\end{definition}

The LERC extends the exact recoverability situation $\pi \in \mathcal{Q}$, as it is satisfied when $\mathcal{Q}_\pi =\mathcal{Q}$, but also in situations where $\mathcal{Q}_\pi$ is larger than $\mathcal{Q}$.

\begin{example}
    If $\mathcal{Q}$ is the family of centered Gaussian distributions with diagonal covariance matrices and $\pi = \mathcal{N}(\mu_\pi, \Sigma_\pi)$ with no restriction on $\mu_\pi$ and $\Sigma_\pi$, then $\pi$ satisfies the LERC with respect to $\mathcal{Q}$. The linear operator is $L_\pi: \omega \rightarrow (0, \diag \omega)$, from $ \mathbb{R}^d$ to $\mathbb{R}^d \times \mathbb{S}^d$, where for $\omega \in \mathbb{R}^d$, $\diag \omega$ denotes the diagonal matrix whose diagonal elements are the elements of $\omega$. In the special case $\mu_\pi=0$, this corresponds to a Gaussian mean-field variational approximation, as the correlation structure of $\pi$ is removed. More details can be found in Appendix~\ref{app:expFam}.
\end{example}

While generalizing exact recoverability, the LERC also has a direct simplifying impact on solving our VI problem, as shown in the next result. 

\begin{proposition}
    \label{prop:LERCConsequences}
    Consider $\pi \in \mathcal{Q}_\pi$ that satisfies the LERC with respect to $\mathcal{Q}$ with $\theta_\pi$ the parameter so that $\pi=q_{\theta_\pi}$.
    \begin{itemize}
        \item [$(i)$] Suppose that Assumptions \ref{assumption:steep} and \ref{assumption:constraint} hold. If $C$ is compact, then Assumption \ref{assumption:interior} is satisfied. Alternatively, if $L_\pi^\top \theta_\pi \in \interior \domain A$, then \ref{assumption:interior} is satisfied and $\omega_* = \proj_C^{A^*}(\nabla A(L_\pi^\top \theta_\pi))$. 
        \item [$(ii)$] Under Assumption \ref{assumption:steep}, $f^D_\pi$ is $1$-strongly-convex and $1$-smooth relatively to $A^*$.
    \end{itemize}
\end{proposition}

Proposition \ref{prop:LERCConsequences} $(i)$ shows that the LERC helps ensuring that Assumption \ref{assumption:interior} is satisfied, which amounts to showing that our VI problem \eqref{eq:VIExpFamConstraint} has a unique well-defined solution $\omega_*$. If $L_\pi^\top \theta_\pi \in \interior \domain A$, this solution can be further characterized. Without constraints ($C = \mathcal{H})$, we have in particular that $\omega_*$ is such that $\theta_* = \nabla A^*(\omega_*) = L_\pi^\top \theta_\pi$. However, this solution may still be hard to find due to the stochasticity of Algorithm~\ref{alg}. Fortunately, Proposition \ref{prop:LERCConsequences} $(ii)$ establishes that the LERC also implies useful properties for the convergence of Algorithm \ref{alg}. Indeed, having $f^D_\pi$ smooth relatively to $A^*$ allows to bound the variance $\sigma_{\eta, N}^2$ using Proposition~\ref{prop:varianceBoundHadrien} and the strong convexity of $f^D_\pi$ relatively to $A^*$ allows to obtain fast convergence rates in Proposition~\ref{prop:cvgceStochMD} and Proposition~\ref{prop:cvgceStochMDDecreasing}.

\section{ILLUSTRATIONS}
\label{sec:Illus}

In this section, we illustrate into more details  two  settings in which  our new convergence results apply, as stated in Propositions \ref{prop:varianceBonnetPrice} and \ref{prop:boundVarianceSubsampling}.

\subsection{Sampling-based gradients over Gaussians}
\label{subsec:BonnetPrice}

In this section, no specific assumption is made on $\pi$ at first but $\mathcal{Q}$ is the family of Gaussian distributions.
In this context, we consider Monte Carlo estimators of the gradients $\nabla_\omega \mathbb{E}_{q_\omega}[\log \pi(X)]$. Several estimators of this quantity have been used in VI, such as score-based estimators \citep{ranganath2014}, or estimators based on the reparametrization trick \citep{kingma2014}. However, we will use another estimator that leverages the fact that $\mathcal{Q}$ is a Gaussian family allowing to obtain better performance \citep{wu2024}.

For $\omega = (\mu, \Sigma + \mu \mu^{\top})$, using theorems from \cite{bonnet1964} and \cite{price1958} and the chain rule (see also \cite{khan2023}), we have
\begin{align*}
    \nabla_{\omega_1} \mathbb{E}_{q_{\omega}}[\log \pi(X)] &= \mathbb{E}_{q_{\omega}}[\nabla \log \pi(X) - \nabla^2 \log \pi(X) \mu]\\
    \nabla_{\omega_2} \mathbb{E}_{q_{\omega}}[\log \pi(X)] &= \frac{1}{2} \mathbb{E}_{q_{\omega}}[\nabla^2 \log \pi(X)].
\end{align*}
Motivated by these formula, the so-called Bonnet and Price estimator $g^N(\omega) = (g^N(\omega)_1, g^N(\omega)_2)$ is built from $N$ samples $X_n \sim q_{\omega}$, $n \in \llbracket 1, N \rrbracket$ by
\begin{eqnarray}
        g^N(\omega)_1 &=& \frac{1}{N} \sum_{n=1}^N \left( \nabla \log \pi(X_n) - \nabla^2 \log \pi(X_n) \mu \right) \nonumber\\
        g^N(\omega)_2 &=& \frac{1}{2N} \sum_{n=1}^N \nabla^2 \log \pi(X_n).
    \label{eq:BonnetPriceEstimator}
\end{eqnarray}
This estimator is unbiased. If $\pi$ is Gaussian, the following proposition states that its associated variance $\sigma_{\eta,N}^2$ is upper-bounded.
\begin{proposition}
    \label{prop:varianceBonnetPrice}
    Suppose that $\pi$ is a Gaussian distribution with covariance $\Sigma_\pi \in \mathbb{S}^d_{>0}$. Then, $g^N$ is such that Assumptions \ref{assumption:wellposedness} and \ref{assumption:estimator} are satisfied and there exists a constant $V > 0$ that depends only on $\Sigma_\pi$, $C$, and $\omega_0$ s.t.~$\sigma_{\eta,N}^2 \leq \frac{V}{N}$ for all $\eta \in (0,1]$.
\end{proposition}

\subsection{Data subsampling}
\label{subsec:BayesianLinearRegression}

In this section, no  assumption on ${\cal Q}$ is made at first, but  $\pi$ has the form $\pi(x | y) \propto p_0(x) \; \prod_{m=1}^M p(y_m | x) $ with  $y= \{y_m\}_{m=1}^M$ and $M$ very large. The prior distribution $p_0$ is in an exponential family ${\cal Q}_0$ and satisfies the LERC with respect to ${\cal Q}$. Let $L_0$, $\Gamma_0$, $A_0$ and $\theta_0$ be respectively the linear transformation, sufficient statistic, log-partition, and natural parameter corresponding to $p_0$, so that
$\mathbb{E}_{q_\omega}\left[\log p_0(X)\right] = \langle \theta_0, \mathbb{E}_{q_\omega}\left[\Gamma_0(X)\right]\rangle  -A_0(\theta_0) = \langle L_0^\top \theta_0, \omega\rangle -A_0(\theta_0)$. 
Equation \eqref{def:lulu} writes
\begin{align*}
    f^D_\pi(\omega) = & A^*(\omega) - \langle L_0^\top \theta_0, \omega\rangle  \\
     & + \sum_{m=1}^M  \mathbb{E}_{q_\omega}\left[\log p(y_m | X)\right] + \textrm{constant}. 
\end{align*}
and 
$  \nabla f^D_\pi(\omega) = \theta - \left( L_0^{\top}\theta_0 + \sum_{m=1}^M \theta_{y_m}(\omega) \right)$,
where 
$$\theta_{y_m}(\omega)  = \nabla \mathbb{E}_{q_\omega}\left[\log p(y_m | X)\right]$$
is assumed to be tractable.
However, for the gradient evaluation, the large sum of $M$ terms is problematic.
To circumvent this issue, subsampling only a subset of the data points at every iteration can be considered, thus creating a cheap stochastic estimator of the optimal solution. 
Consider a uniform random variable $U \sim \mathcal{U}_M$. The gradient can be rewritten as 
    $\nabla f^D_\pi(\omega) = \theta - \left( L_0^{\top}\theta_0 + M \mathbb{E}_{\mathcal{U}_M}\left[\theta_{y_U}(\omega) \right] \right)$,
and the expectation  approximated by sampling a smaller batch of $N$ data points uniformly, yielding the estimator $g^N(\omega)$ of $\nabla \mathbb{E}_{q_{\omega}}[\log \pi(X)]$ defined by
\begin{equation}
    \label{eq:subsamplingEstimator}
    g^N(\omega) = L_0^\top \theta_0 + \frac{M}{N} \sum_{n=1}^N \theta_{y_{U_n}}(\omega),
\end{equation}
where $U_n \sim \mathcal{U}_M$ for $n \in \llbracket 1, N \rrbracket$. This estimator can then be plugged in 
Algorithm \ref{alg} to get an approximation of the optimal $\theta^*$.

\paragraph{Bayesian linear regression} As an illustration of the above situation, we investigate a specific Bayesian regression task with the subsampling estimator  \eqref{eq:subsamplingEstimator}. We consider the same setting as \cite{wu2024}, that is Bayesian linear regression with a centered Gaussian prior, Gaussian likelihood, and Gaussian approximating family ${\cal Q}$. More specifically,

$X \sim p_0 = \mathcal{N}(0,\Sigma_0)$ with $\theta_0 = -\frac{1}{2}\Sigma_0^{-1}$ and $Y \in \mathbb{R}$ with
\begin{equation}
    \label{eq:linearRegressionLikelihood}
    Y | X=x \sim \mathcal{N}(x^\top z, \sigma^2),
\end{equation}
with $z$ a vector of fixed covariates values associated to $Y$.
Given $M$ data points $\{y_m\}_{m=1}^M$ and associated covariates $\{z_m\}_{m=1}^M$, the posterior $\pi (x | y_1 \ldots y_M)$  is a Gaussian distribution with natural parameter
\begin{equation}
    \label{eq:BayesianLinearRegressionTarget}
    \theta_\pi = \left( \frac{1}{\sigma^2} \sum_{m=1}^M y_m z_m, \theta_0 - \frac{1}{2 \sigma^2} \sum_{m=1}^M z_m z_m^\top \right).
\end{equation}
Assume  ${\cal Q}\! =\!\{q_\omega, \; \omega\in\domain A^*\}$ is the set of Gaussian distributions  with expected parameters $\omega \!\!= \!\!(\omega_1, \omega_2)\!\! =\! \!(\mu, \Sigma + \mu \mu^\top\!)$. It follows that $p_0$ satisfies the LERC since it is Gaussian and thus in $\mathcal{Q}$. For each $y_m$, $\theta_{y_m}\!(\omega)~=~(\nabla_{\omega_1}\! \mathbb{E}_{q_\omega}\!\!\left[\log p(y_m | X)\right]\!,\nabla_{\omega_2} \!\mathbb{E}_{q_\omega}\!\!\left[\log p(y_m | X)\right])$ is available in closed form:
\begin{align*}
\theta_{y_m}(\omega)
&= \frac{1}{2 \sigma^2}\left(2y_m z_m, - z_m z_m^\top\right).
\end{align*}
It follows that $g^N$ in \eqref{eq:subsamplingEstimator} writes
\begin{equation*}
    g^N(\omega) = \frac{M}{N2 \sigma^2} \sum_{n=1}^N  \left( 2 y_{U_n} z_{U_n},  -  \frac{1}{M}\sigma^2 \Sigma_0^{-1} -  z_{U_n}z_{U_n}^\top \right). 
\end{equation*}
We are in the required setting to define $g^N$ as in \eqref{eq:subsamplingEstimator} and  this estimator takes values in $\interior \domain A$.

\begin{proposition}
    \label{prop:boundVarianceSubsampling}
    When $\pi = q_{\theta_\pi}$ as in Eq.~\eqref{eq:BayesianLinearRegressionTarget} and when $g^N$ is defined as in \eqref{eq:subsamplingEstimator}, Assumptions \ref{assumption:wellposedness} and \ref{assumption:estimator} are verified and there exists a constant $V > 0$ which depends only on $\Sigma_0$ and on the data such that $\sigma_{\eta,N}^2 \leq \frac{V}{N}$ for any $\eta \in (0,1]$.
\end{proposition}

\subsection{Convergence of Algorithm~\ref{alg}}

Both previous Propositions allow to conclude about the convergence of Algorithm~\ref{alg}.
\begin{corollary}
    \label{corollary:BayesianLinearRegression}
    Suppose that Algorithm~\ref{alg} is run in either Propositions \ref{prop:varianceBonnetPrice} or \ref{prop:boundVarianceSubsampling} settings, then
    \begin{itemize}
        \item [$(i)$] If $\eta_t \equiv \eta \in (0,1]$ and $N_t \equiv N \in \mathbb{N}_{>0}$, then Algorithm~\ref{alg} yields iterates $\{ \omega_t \}_{t\in \mathbb{N}}$ s.t.~for any $T \in \mathbb{N}$,
        \begin{equation*}
            \mathbb{E}\left[ d_{A^*}(\omega^*, \omega_T) \right] \leq (1 - \eta)^T d_{A^*}(\omega^*, \omega_0) + \frac{\eta V}{N}.
        \end{equation*}
        \item [$(ii)$] If $\eta_t \equiv \eta \in (0,1]$ and $N_t = (t+1)^\gamma$,$\gamma>0$, then Algorithm~\ref{alg} yields iterates $\{ \omega_t \}_{t\in \mathbb{N}}$ s.t.~for any $T \in \mathbb{N}$,
        \begin{equation*}
            \mathbb{E} \left[ d_{A^*}(\omega_*, \omega_T) \right] = \mathcal{O}\left( \frac{1}{T^\gamma} \right)
        \end{equation*}
        \item [$(iii)$]
        If $\eta_t = \frac{1}{t/2+1}$, then Algorithm~\ref{alg} yields iterates $\{ \omega_t \}_{t \in \mathbb{N}}$ s.t.~for any $T \in \mathbb{N}$,
        \begin{equation*}
            \mathbb{E}\left[ d_{A^*}(\omega^*, \omega_{T+1}) \right] = \mathcal{O}\left( \frac{1}{T^2} \sum_{t=0}^T \frac{1}{N_t}\right).
        \end{equation*}
    \end{itemize}
\end{corollary}

\section{RELATED WORK}
\label{sec:related}

The convergence of stochastic NGVI has seldom been studied. Existing work exploits the connection between natural gradient descent and mirror descent when optimising over an exponential family (see Section~\ref{sec:background}). \cite{khan2016} gave convergence rates to stationarity assuming strong convexity of $A$ and a bounded variance, which may fail on some parts of the search space \citep{guilmeau2025, domke2019}. \cite{wu2024} exploited a variance-like quantity proposed by \cite{hanzely2021}, which is provably larger (possibly infinite) than $\sigma_{\eta,N}^2$ \citep{hendrikx2024}. When $\pi \in \mathcal{Q}$, they showed $\mathcal{O}\left( \frac{1}{T}\right)$ convergence to the minimizer using decreasing step sizes, fixed sample/batch sizes, and Polyak-Ruppert averaging \citep{polyak1992}. They demonstrated that their considered variance is finite in a conjugate Bayesian linear regression setting. \cite{sun2025} also used the variance introduced by \cite{hanzely2021} to derive convergence rates for convex and non-convex VI problems, covering for instance logistic regression, but which are restricted to mean-field Gaussian approximating families. Also, \cite{sun2025} used a projection step to a compact constraint set that is necessary for convergence.
In Proposition~\ref{prop:cvgceStochMDDecreasing}, we show a similar rate to \cite{wu2024} without averaging. 
In addition, we provide in Propositions~\ref{prop:cvgceStochMD}-\ref{prop:cvgceStochMDDecreasing} additional possibilities to combine fixed or decreasing step sizes and fixed or increasing sample/batch sizes. These allow to further accelerate convergence and may be closer to what practitioners use. Our result apply for any exponential family $\mathcal{Q}$ and constraint set $C$ (possibly, $C = \mathcal{H}$) and not only for Gaussian mean-field approximations and compact $C$, although we do require convexity properties. In Section \ref{sec:Illus}, we apply these results to various settings, covering in particular the conjugate Bayesian linear regression task studied by \cite{wu2024}.

Our Corollary \ref{corollary:BayesianLinearRegression} $(i)$ showcases geometric rates with a bias decreasing inversely with the step size. Similar results were obtained by \cite{lambert2022} and \cite{diao2023} for VI algorithms leveraging the Bures-Wasserstein geometry over Gaussian distributions and by \cite{domke2023} for location-scale VI. However, they considered a sample size $N=1$, while our result highlight the inverse dependence on $N > 0$ and the case of increasing $N_t$ is covered in Corollary \ref{corollary:BayesianLinearRegression} $(ii)$.

Although schedules such as in Propositions \ref{prop:cvgceStochMD} and \ref{prop:cvgceStochMDDecreasing} were seldom considered in the VI literature, joint conditions on the step sizes and the sample/batch sizes have been considered within the optimization community, for instance by \cite{ghadimi2016} for projected gradient descent or by \cite{xiao2021} for Bregman proximal gradient descent. Our results compliment theirs by using a tighter variance definition and exploiting the relative strong convexity of the objective to yield precise rates. Empirical performance is moreover investigated in Section \ref{sec:experiments}.

Our analysis allows for projection steps (explicit projection operators can be found in Appendix~\ref{app:projections} and additional experiments in Appendix \ref{app:additionalExperiments}). Projections have been deemed necessary in the context of stochastic (Euclidean) gradient descent algorithm for VI over location-scale families by \cite{domke2023} and \cite{kim2023} but also in the context of natural gradient descent by \cite{sun2025}. Note that in our work, projection is possible, but not mandatory. For natural gradient descent, our projection step can also be seen as an alternative to the Riemannian construction proposed by \cite{lin2020} and as a way to guarantee the assumptions made by \cite{kumar2025}.

\cite{margossian2025} identified conditions on $\pi$ under which location-scale VI (generally a non-convex problem) admits no local minimizer. These conditions do not force $\pi$ to belong to the family of approximating distributions, but still require some similarity between $\pi$ and its approximations. Similarly, our LERC from Definition~\ref{def:linearEnlargement} requires $\pi$ to be close to $\mathcal{Q}$ and ensures that the VI problem admits a unique solution.

\section{NUMERICAL EXPERIMENTS}
\label{sec:experiments}

Numerical illustrations are provided in the settings of Section \ref{sec:Illus}: $\mathcal{Q}$ is the family of Gaussian distributions and gradients are estimated either with Bonnet and Price estimators or using subsampling. Projections are not used in this section ($C = \mathcal{H})$. We consider different schedules, namely constant $\eta_t \equiv \eta \in (0,1]$ as in Proposition \ref{prop:cvgceStochMD} and decreasing $\eta_t = \frac{1}{t/2+1}$ as in Proposition \ref{prop:cvgceStochMDDecreasing}, as well as constant sample/batch size $N_t \equiv N \in \mathbb{N}_{>0}$ and increasing sample/batch size $N_t = t+1$. Comparing such schedules in terms of iterations can be misleading, as  iterations may not involve the same number of data points. Therefore, we also provide plots with respect to the total computational budget used until iteration $t$, that is, $\sum_{\tau=0}^t N_\tau$. 

All experiments were run on a personal laptop with 7.6 GB RAM and 8 Intel Core i5-8265U cores. Code can be found at \href{https://github.com/tGuilmeau/Projected_Stochastic_NGVI/}{this Github repository}.

\paragraph{Bonnet and Price gradients}

We test Algorithm~\ref{alg} with the different proposed schedules in the setting of Proposition \ref{prop:varianceBonnetPrice}, with $\pi$ a synthetic Gaussian target in dimension $d=10$ with condition number $\kappa = 10^2$ generated as in \cite{more1989}, and $g^N$ is computed  using the estimator \eqref{eq:BonnetPriceEstimator}.
Figure~\ref{fig:BonnetPriceIterations} confirms the result of Propositions \ref{prop:cvgceStochMD} and \ref{prop:cvgceStochMDDecreasing}: for constant step sizes and sample sizes, the iterates converge geometrically to a neighborhood of $\omega_*$, but if  either step sizes decrease or sample sizes increase, the iterates converge to $\omega_*$ at a sub-geometric rate. In particular, the convergence is the fastest with decreasing step sizes and increasing sample sizes.
Figure~\ref{fig:BonnetPriceComputationBudget} shows that in terms of computational budget, it may be beneficial to keep a fixed sample size $N_t \equiv N$, as the schedule with decreasing $\eta_t$ and fixed $N_t$ achieves the best performance within the given computational budget.

\paragraph{Data subsampling}

We now test Algorithm~\ref{alg} in the setting of Proposition \ref{prop:boundVarianceSubsampling}, where  $\pi$ is a Bayesian linear regression posterior, and $g^N$ is computed by subsampling data points as in \eqref{eq:subsamplingEstimator}. We use the \cite{dataset} dataset \citep{kaya2019} with Gaussian prior $\mathcal{N}(0, 5 I)$ and Gaussian likelihood as in \eqref{eq:linearRegressionLikelihood} with $\sigma^2 = 1$. Here, $M = 36,733$ and $d=9$.

Figure~\ref{fig:subSamplingIteration} and \ref{fig:subSamplingBudget} showcase results for subsampling-based estimators that are similar to the results obtained with estimators based on Bonnet's and Price's theorems. Namely, quick convergence to a neighborhood of the solution when fixed step and batch sizes are used, and convergence to the solution when either step sizes decrease or batch sizes increase. Again, note that decreasing step sizes and increasing batch sizes may be faster in terms of iterations, but this is not so clear anymore in terms of computational budget. Note also that schedules where $N_t \rightarrow +\infty$ may be unrealistic in practice, as batch sizes would reach the total number of data points. 

\begin{figure}[H]
    \centering
    \begin{subfigure}[t]{.49\linewidth}
        \centering
        \includegraphics[width=\linewidth]{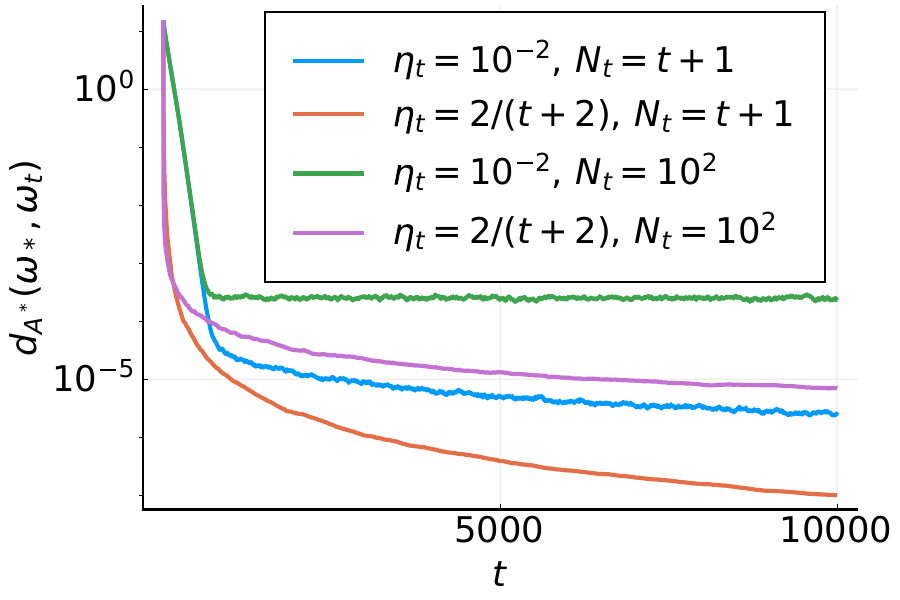}
        \caption{wrt iterations.}
        \label{fig:BonnetPriceIterations}
    \end{subfigure}\hfill
    \begin{subfigure}[t]{.49\linewidth}
        \centering
        \includegraphics[width=\linewidth]{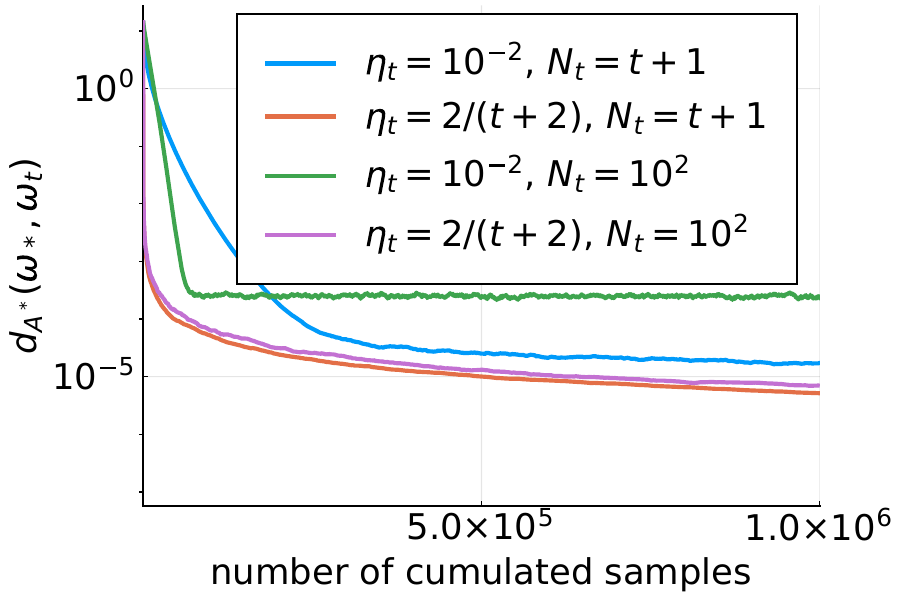}
        \caption{wrt  numbers of samples.}
        \label{fig:BonnetPriceComputationBudget}
    \end{subfigure}
    \centering
    \begin{subfigure}[t]{.49\linewidth}
        \centering
        \includegraphics[width=\linewidth]{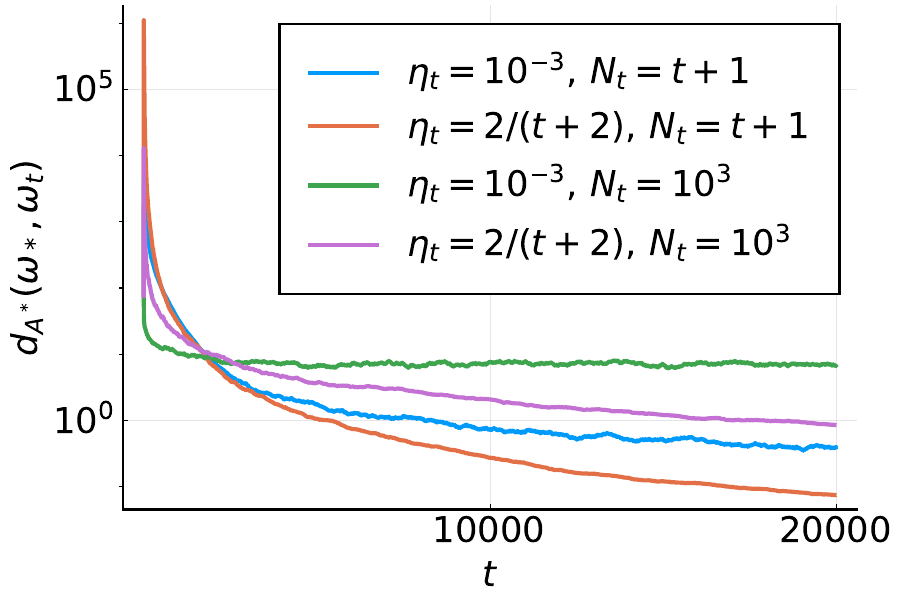}
        \caption{wrt iterations.}
        \label{fig:subSamplingIteration}
    \end{subfigure}
     \hfill 
     \begin{subfigure}[t]{.49\linewidth}
    \centering
        \includegraphics[width=\linewidth]{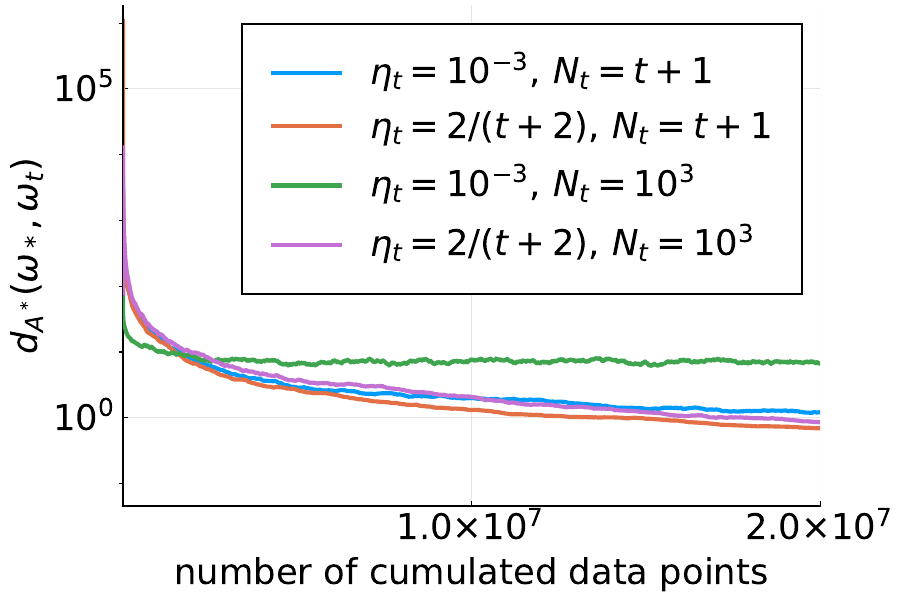}
        \caption{wrt  numbers of samples.}
        \label{fig:subSamplingBudget}
    \end{subfigure} 
    \caption{Mean Bregman divergence between current and optimal parameters, over $100$ runs, for different NGVI schedules in Prop.\ref{prop:varianceBonnetPrice} (a,b) and  \ref{prop:boundVarianceSubsampling} (c,d) settings.}
    \label{fig:BonnetPrice}
\end{figure}

\paragraph{Logistic regression} We then consider a logistic regression task  to test Algorithm \ref{alg} in a setting when $\pi$ does not satisfy the LERC with respect to $\mathcal{Q}$. Although $\pi$ does not satisfy the LERC,  $\log \pi$ is still convex and twice-differentiable, implying in particular that \ref{assumption:wellposedness} holds (see Appendix \ref{app:proof3}). The logistic regression setting is tested in dimension $d=5$. $M=100$ data points $\{z_m\}_{m=1}^M$ are generated uniformly in $[-5,5]^d$, and for each $m \in \llbracket 1,M \rrbracket$, $y_m$ follows a Bernoulli distribution with success probability equal to $1 / (1 + \exp(-\langle x_*, z_m \rangle))$ where $x_* \in \mathbb{R}^d$ with all its components being equal to $5$.
Algorithm \ref{alg} is run with $\mathcal{Q}$ set to the family of Gaussian distributions with diagonal covariances. 
The $\text{KL}(q_\omega, \pi)$  evolution is monitored equivalently by  the ELBO. A higher ELBO translates to a lower KL \citep{zhang2019}. The exact ELBO is approximated by a Monte Carlo sum using $100$ samples.  
Figure \ref{fig:logReg} depicts this Monte-Carlo ELBO averaged over 50 runs, for different sample and step size schedules. Figure \ref{fig:logReg} shows a slightly different picture from the results obtained in Figure \ref{fig:BonnetPrice}. The best ELBO in the transient regime is not reached by the schedule with increasing sample and decreasing step size anymore. Note that our theoretical results only describe the convergence of the Bregman divergence $d_{A^*}$ between the current iterate and the minimizer.

\begin{figure}[H]
    \centering
    \begin{subfigure}[t]{.49\linewidth}
        \centering
        \includegraphics[width=\linewidth]{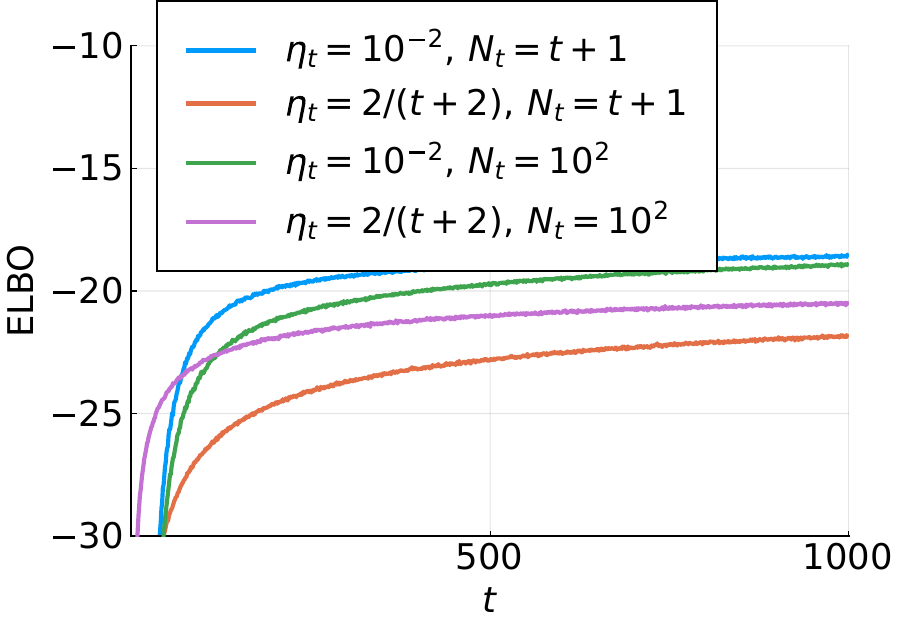}
        \caption{wrt iterations.}
        \label{fig:logRegIterations}
    \end{subfigure}\hfill
    \begin{subfigure}[t]{.49\linewidth}
        \centering
        \includegraphics[width=\linewidth]{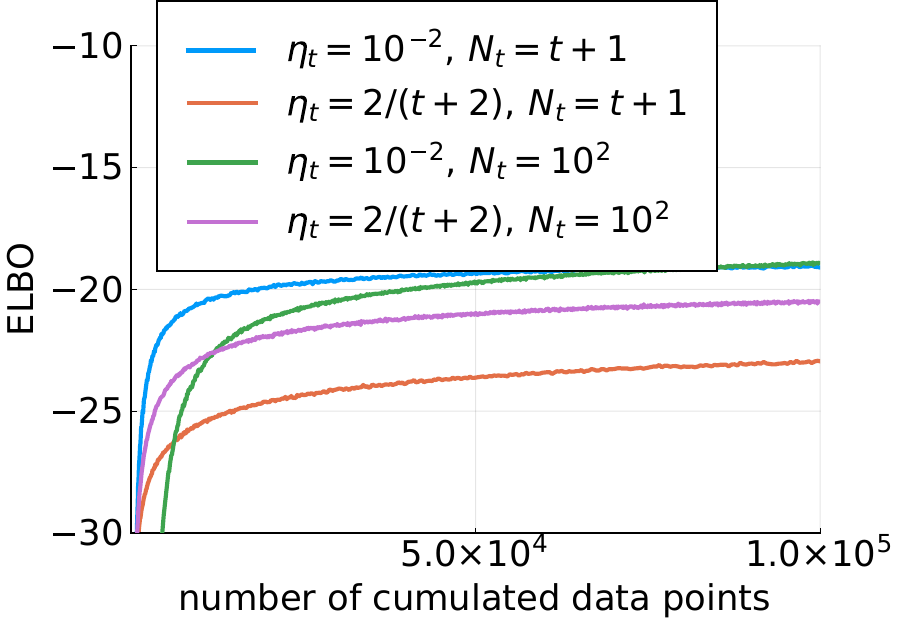}
        \caption{wrt  numbers of samples.}
        \label{fig:logRegComputationBudget}
    \end{subfigure}
    \caption{Logistic regression: Average ELBO over $50$ runs, for different NGVI schedules.}
    \label{fig:logReg}
\end{figure}

\paragraph{Additional experiments} Additional experiments are deferred to Appendix \ref{app:additionalExperiments}. In particular, the impact of potential projection steps is studied on a logistic regression task and on a robust linear regression task where log-convexity of the Bayesian posterior fails.

\section{CONCLUSION}
\label{sec:conclusion}

In this work, we have proposed new non-asymptotic convergence rates for projected stochastic mirror descent algorithms, which may be of independent interest, and showed that they can be applied to NGVI, including in the case of posterior distributions that cannot be exactly recovered. In particular, our results allow combinations of fixed or decreasing step sizes and fixed or increasing sample/batch sizes. These results extend existing results for NGVI to a wider range of settings and explicit the influence of the choice of step sizes and sample/batch sizes schedules.

Our analysis may apply to slightly extended settings. For instance, generalisations of the exponential families, such as the $q$-exponential \citep{amari2011} and $\lambda$-exponential \citep{wong2022, guilmeau2024} families which include heavy-tailed distributions, or such as mixtures \citep{lin2019, nguyen2025}, represent interesting perspectives to extend our analysis.

\subsubsection*{Acknowledgements}

We are grateful to the reviewers for their constructive comments that helped improve the quality of the paper.

\bibliography{references}

\clearpage
\onecolumn
\appendix

\section{EXPONENTIAL FAMILIES}
\label{app:expFam}

\subsection{Steep exponential families}

We recall the notion of steepness from \cite{barndorff-nielsen2014}, which we will use in our proofs.

\begin{definition}
    \label{def:steepness}
    Suppose that $\mathcal{Q}$ is an exponential family with log-partition function $A$ that is differentiable on $\interior \domain A$. We say that $\mathcal{Q}$ is steep if for any sequence $\{ \theta_t \}_{t \in \mathbb{N}}$ such that $\theta_t \in \interior \domain A$ for any $t \in \mathbb{N}$ and $\theta_t \rightarrow \theta$ which belongs to the boundary of $\domain A$, then $\| \nabla A(\theta_t) \| \rightarrow +\infty$.
\end{definition}

It has been established in \cite[Theorem 8.2]{barndorff-nielsen2014} that if $\domain A$ is open, then $\mathcal{Q}$ is steep.

In this section, we choose the base measure $\nu$ to be the Lebesgue measure on $\mathbb{R}^d$ multiplied by $(2 \pi)^{-\frac{d}{2}}$.

\subsection{Gaussian distributions with full covariance}

\begin{proposition}
    \label{prop:fullGaussiansExpFam}
    The family of Gaussian distributions on $\mathbb{R}^d$ forms an exponential family with sufficient statistic $\Gamma : \mathbb{R}^d \rightarrow \mathbb{R}^d \times \mathbb{S}^d$ with $\Gamma(x) = (x, x x^\top)$. For each $\mu \in \mathbb{R}^d$ and $\Sigma \in \mathbb{S}_{> 0}^d$, the density of $\mathcal{N}(\mu, \Sigma)$ can be written under the form $q_\theta$ outlined in Eq.~\eqref{eq:expFamilyDensity} with
    \begin{align*}
        \theta = (\theta_1, \theta_2) &= \left( \Sigma^{-1}\mu, -\frac{1}{2} \Sigma^{-1}  \right),\\
        A(\theta) &= -\frac{1}{4} \theta_1^\top \theta_2^{-1} \theta_1 - \frac{1}{2}\logdet(-2\theta_2).
    \end{align*}
    We have $\domain A = \mathbb{R}^d \times \mathbb{S}^d_{<0}$, where $\mathbb{S}^d_{<0}$ denotes the semi-definite negative matrices. Further, the resulting exponential family is steep.
\end{proposition}

\begin{proof}
    Denote by $q_{\mu,\Sigma}$ the density of $\mathcal{N}(\mu, \Sigma)$ with respect to $\nu$. Then, at any $x \in \mathbb{R}^d$, we have
    \begin{align*}
        q_{\mu,\Sigma}(x) &= \exp \left(- \frac{1}{2}(x - \mu)^\top \Sigma^{-1} (x - \mu) - \frac{1}{2} \logdet(\Sigma) \right)\\
        &= \exp \left( \langle -\frac{1}{2}\Sigma^{-1}, xx^\top \rangle + \langle \Sigma^{-1} \mu, x \rangle - \frac{1}{2} \mu^\top \Sigma^{-1} \mu  - \frac{1}{2} \logdet(\Sigma) \right).
    \end{align*}
    We can readily derive the expression of $\theta$, $A$, and $\domain A$ from here. The steepness property comes from recognising that $\domain A$ is open and applying \cite[Theorem 8.2]{barndorff-nielsen2014}.    
\end{proof}

Proposition~\ref{prop:fullGaussiansExpFam} indicates in particular that Gaussian distributions (as an exponential family) enjoy the properties outlined in Proposition~\ref{prop:expFamilyProperties}. In particular, the parameters $\theta =( \Sigma^{-1} \mu, -\frac{1}{2} \Sigma^{-1})$ are in bijection with the parameters $\omega = (\mu, \Sigma + \mu \mu^\top)$.

\subsection{Gaussian distributions with diagonal covariance matrices}
\label{subsec:diagonalGaussian}

 We denote the Hadamard product by $\bullet$.

\begin{proposition}
    The family of Gaussian distributions on $\mathbb{R}^d$ with diagonal covariance matrices forms an exponential family with sufficient statistic $\Gamma: \mathbb{R}^d \rightarrow \mathbb{R}^d \times \mathbb{R}^d$ with $\Gamma(x) = (x, x \bullet x)$. For each $\mu \in \mathbb{R}^d$ and $\sigma^2 \in \mathbb{R}^d_{> 0}$, the density of $\mathcal{N}(\mu, \diag \sigma^2)$ can be written under the form $q_\theta$ outlined in Eq.~\eqref{eq:expFamilyDensity} with
    \begin{align*}
        \theta = (\theta_1, \theta_2) &= \left(\left(\frac{1}{\sigma^2}\right)\bullet\mu, - \frac{1}{2}\frac{1}{\sigma^2}\right),
    \end{align*}
    with the inversion operations being taken point-wise, that is, $\left(\frac{1}{\sigma^2}\right)_i = \frac{1}{(\sigma^2)_i} $ for $i \in \llbracket 1,d \rrbracket$ with $(\cdot)_i$ denoting the i$^{th}$ element of a vector. Moreover, we have
    \begin{equation*}
        A(\theta) = -\frac{1}{4} \sum_{i=1}^d \frac{(\theta_{1}^2)_i}{(\theta_{2})_i} - \frac{1}{2} \sum_{i=1}^d \log(-2 (\theta_{2})_i),
    \end{equation*}
    with $\domain A = \mathbb{R}^d \times \mathbb{R}^d_{<0}$ and the resulting exponential family is steep.
\end{proposition}

\begin{proof}
    The proof can be obtained using the same steps as in the proof of Proposition~\ref{prop:fullGaussiansExpFam}.
\end{proof}

\begin{proposition}
    The family of centered Gaussian distributions with diagonal covariance matrices on $\mathbb{R}^d$ forms an exponential family with sufficient statistic $\Gamma : \mathbb{R}^d \rightarrow \mathbb{R}^d$ with $\Gamma(x) = x \bullet x$. For any $\sigma^2 \in \mathbb{R}^d_{>0}$, the density of $\mathcal{N}(0, \diag \sigma^2)$ can be written as in Eq.~\eqref{eq:expFamilyDensity} with $\theta = -\frac{1}{2} \frac{1}{\sigma^2}$ and $A(\theta) = - \frac{1}{2} \sum_{i=1}^d \log(-2 \theta_i)$. We have that $\domain A = \mathbb{R}_{<0}^d$ and the resulting exponential family is steep.
\end{proposition}

\begin{proof}
    The proof of this proposition can also be obtained using the same steps as in the proof of Proposition~\ref{prop:fullGaussiansExpFam}.
\end{proof}

We now check that the family of Gaussian distributions over $\mathbb{R}^d$ with full covariance matrices linearly enlarges the family of centered Gaussian distributions with diagonal covariance matrices over $\mathbb{R}^d$. Indeed, consider $q_\theta$ in the latter family. Then, $\mathbb{E}_{q_\theta}[X \bullet X] = \sigma^2$ while $\mathbb{E}_{q_\theta}\left[ X \right] = 0$ and $\mathbb{E}_{q_\theta}[X X^\top] = \diag \sigma^2$. Hence, we have that $\mathbb{E}_{q_\theta}\left[ (X, XX^\top)\right] = (0, \diag \mathbb{E}_{q_\theta}\left[ X \bullet X \right])$, which shows the LERC.

\section{PROOFS OF SECTION \ref{sec:background}}
\label{app:proofs}

Most proofs of Section \ref{sec:background} are made easier by using the notion of Legendre functions, which we define now.

\begin{definition}
    \label{def:LegendreFunction}
    Consider a proper convex function $h : \mathcal{H} \rightarrow \mathbb{R} \cup \{ \infty\}$. 
    
    Then, $h$ is called essentially smooth if $\interior \domain h \neq \emptyset$, $h$ is differentiable on $\interior \domain h$, and if for every sequence $\{ \omega_t \}_{t \in \mathbb{N}}$ such that $\omega_t \in \interior \domain h$ for any $t \in \mathbb{N}$ and $\omega_t \rightarrow \omega$ in the boundary of $\domain h$, we have the limit $\| \nabla h(\omega_t) \| \rightarrow +\infty$.
    
    If $h$ proper, lower-semicontinuous, strictly convex, and essentially smooth, then $h$ will be called Legendre.
\end{definition}

Legendre functions benefit from \cite[Theorem 26.5]{rockafellar1970}, which states the following.

\begin{proposition}
    \label{prop:legendreProperties}
    Consider a Legendre function $h$. Then, we have the following:
    \begin{itemize}
        \item [$(i)$] $h$ is Legendre if and only if its convex conjugate $h^*$ is Legendre.
        \item [$(ii)$] $\nabla h$ is a bijection from $\interior \domain h$ to $\interior \domain h^*$ and its inverse is $\nabla h^*$.
    \end{itemize}
\end{proposition}

In this work, we use the notion of steepness from Definition~\ref{def:steepness}. Remark that the steepness of the exponential family $\mathcal{Q}$ is equivalent to the essential smoothness (defined in Definition~\ref{def:LegendreFunction}) of its log-partition  function $A$.

Now that we have laid out some preliminary notions, we turn to the proofs of Section \ref{sec:background}.

\begin{lemma}
    \label{lemma:propertiesA}
    Suppose that $\interior \domain A \neq \emptyset$. Then, $A$ is proper, lower semi-continuous, strictly convex, and all the partial derivatives of $A$ exist on $\interior \domain A$. Furthermore, if $\mathcal{Q}$ is steep, then $A$ is Legendre.
\end{lemma}

\begin{proof}
    The properness of $A$ comes from the definition of $A$. The lower semi-continuity and convexity come from \cite[Theorem 1.13]{brown1986}. The derivability result comes from \cite[Theorem 8.1]{barndorff-nielsen2014}. Finally, if $\mathcal{Q}$ is steep, then $A$ is essentially smooth. Combined with the previous properties, we see that stepness implies that $A$ is Legendre.
\end{proof}

\subsection{Proof of Proposition~\ref{prop:expFamilyProperties} }
\begin{proof}
    For this proof, we will leverage the results from Lemma \ref{lemma:propertiesA}.

    $(i)$ We use the first part of Lemma \ref{lemma:propertiesA}. Next, denoting by $\mathbb{V}_{q_\theta}[\Gamma(X)]$ the variance of  $\Gamma(X)$, we have from \cite[Chapter 8, Eq.~(12)-(13)]{barndorff-nielsen2014} that
    \begin{align*}
        \nabla A(\theta) &= \mathbb{E}_{q_\theta}[\Gamma(X)] \qquad \qquad \mbox{and} &  \nabla^2 A(\theta) = \mathbb{V}_{q_\theta}[\Gamma(X)].
    \end{align*}
    Then, the result comes by checking that $\nabla_\theta \log q_\theta(x) = \Gamma(x) - \mathbb{E}_{q_\theta}[\Gamma(X)]$, thus showing that $\nabla^2 A(\theta) = \mathbb{E}_{q_\theta}\left[ (\nabla_\theta \log q_\theta(X))(\nabla_\theta \log q_\theta(X))^\top  \right]$ which is the Fisher information matrix.

    $(ii)$ $A$ is Legendre from Lemma \ref{lemma:propertiesA}, so $A^*$ is also Legendre from Proposition~\ref{prop:legendreProperties}. This shows that $A^*$ is proper, lower semi-continuous, strictly convex, and differentiable on $\interior \domain A^*$. For any $\omega \in \interior \domain A^*$, we can rewrite $A^*(\omega) = \langle \theta, \nabla A(\theta) \rangle - A(\theta)$ where $\omega \in \nabla A^*(\theta)$ using the Fenchel-Young equality \cite[Prop.~16.9]{bauschke2011}. It follows that 
    $A^*(\omega) = \langle \theta, \mathbb{E}_{q_\theta}[\Gamma(X)] \rangle - A(\theta) = \mathbb{E}_{q_\theta}[\langle \theta, \Gamma(X) \rangle - A(\theta)]= \mathbb{E}_{q_\theta}[\log q_\theta(X)]$ as desired. 

    $(iii)$ This result comes from $A$ being Legendre (Lemma \ref{lemma:propertiesA}) and the result of Proposition~\ref{prop:legendreProperties} $(ii)$.

    $(iv)$ We first have that $d_A(\theta, \theta') = d_{A^*}(\omega', \omega)$ from \cite[Theorem 3.7 (v)]{bauschke1997}. Then, the equality with the KL divergence comes from \cite{nielsen2010}.
\end{proof}

\subsection{Proof of Proposition~\ref{prop:BregmanProjection}}

\begin{proof}
    We first show that the domain of $\proj_C^{h}$ is $\interior \domain h$, meaning that $\proj_C^h(\omega) \neq \emptyset$ if and only if $\omega \in \interior \domain h$. This is done by remarking that $h$ is Legendre so it is essentially strictly convex and that $C$ is closed convex with $C \cap \domain h \neq \emptyset$. Thus, we apply \cite[Proposition 3.31 (vi)]{bauschke2003} and obtain the result.

    We now show that $\proj_C^{h}(\omega) \subset \interior \domain h$ for any $\omega \in \interior \domain h$. This is done by remarking that $C \cap \interior \domain h \neq \emptyset$ and that $h$ is Legendre so essentially smooth, and using \cite[Proposition 3.33 (v)(b)]{bauschke2003}.

    We now demonstrate that $\proj_C^{h}$ is single-valued on $\interior \domain h$. To do so, remark that we can apply \cite[Proposition 3.32 (ii)(d)]{bauschke2003} due to our previous point and the fact that $h$ restricted to $\interior \domain h$ is strictly convex.
\end{proof}

Note that when $\proj_C^h(\omega)$ is single-valued, we will consider that $\proj_C^h(\omega)$ is a point of $\mathcal{H}$ instead of being a singleton of $\mathcal{H}$. In the paper, the above result will typically be applied with $h= A^*$.

\subsection{Proof of Proposition~\ref{prop:varianceBoundHadrien}}
\begin{proof}
    We can proceed as in the proof of \cite[Prop.~2.4]{hendrikx2024} to obtain that for any $\omega \in D$, we have
    \begin{equation*}
        f(\omega_*) - f_\eta(\omega) \leq \frac{1}{\eta} \mathbb{E}\left[d_h(\bar{\omega}_+, \omega_+) \right].
    \end{equation*}
    We then obtain our result by taking the supremum over $D$.
\end{proof}

\section{PROOFS OF SECTION \ref{sec:presentation}}

\label{app:proof3}

\subsection{Necessary conditions ensuring Assumption \ref{assumption:wellposedness}}

Assumption \ref{assumption:wellposedness} requires the iterates generated by Algorithm \ref{alg} to stay in $\interior \domain A^*$. Proposition \ref{prop:sufficientWellPosedness} below  provides necessary conditions that ensure that any iterate in $\interior \domain A^*$ will generate a new iterate in $\interior \domain A^*$ as well. Then, this result can be applied recursively to satisfy Assumption \ref{assumption:wellposedness}.

\begin{proposition}
    \label{prop:sufficientWellPosedness}
    Suppose that Assumptions \ref{assumption:steep} and \ref{assumption:constraint} hold and consider an iterate $\omega_t \in \interior \domain A^*$, a step size $\eta_t > 0$, as well as the resulting updated point $\omega_{t+1}$ defined by the recursion of Algorithm \ref{alg}. Then, we have the following:
    \begin{itemize}
        \item [$(i)$] If for any $\omega \in \interior \domain A^*$, $N \in \mathbb{N}_{>0}$, $g^N(\omega) \in \interior \domain A^*$ almost surely, then $\omega_{t+1} \in \interior \domain A^*$ almost surely for any step size $\eta_t \in [0,1]$.
        \item [$(ii)$] There exists $\bar{\eta}$ such that if $\eta_t \in (0,\bar{\eta})$, $\omega_{t+1} \in \interior \domain A^*$ almost surely.
    \end{itemize}
\end{proposition}
Note that using similar arguments, point $(i)$ holds with $g^N(\omega) \in \domain A^*$ almost surely as long as $\eta_t < 1$.
\begin{proof}
    In order to prove $(i)$ and $(ii)$, we will each time prove that $\omega_+ \in \interior \domain A^*$ almost surely. Then, the result on $\omega_{t+1}$ will come from applying Proposition \ref{prop:BregmanProjection}.

    $(i)$ Recall that $(\theta_{t})_+ = \nabla A^*((\omega_{t})_+)$ and  $(\theta_{t})_+ = (1-\eta_t) \theta_t + \eta_t g^N(\omega_t)$, with $\eta_t \in (0,1]$. Since $\theta_t = \nabla A^*(\omega_t)$ with $A^*$ Legendre and $\omega_t \in \interior \domain A^*$, $\theta_t \in \interior \domain A$. As this is also assumed to be the case of $g^N(\omega_t)$, we obtain $(\theta_{t})_+ \in \interior \domain A$ from the convexity of this latter set, and by the Legendre property of $A$ and $A^*$, $(\omega_{t})_+ = \nabla A((\theta_{t})_+)$ is in $\interior \domain A^*$.

    $(ii)$ Similarly as in the proof of $(i)$, we have that $(\theta_{t})_+ = (1-\eta_t) \theta_t + \eta_t g^N(\omega_t)$, with $\theta_t \in \interior \domain A$. Since this latter set is open, there exists an open ball of radius $r > 0$, centered on $\theta_t$, that is included in $\interior \domain A^*$. We can also compute
    \begin{equation*}
        \| (\theta_{t})_+ - \theta_t \| \leq \eta_t \| g^N(\omega_t) - \theta_t \|.
    \end{equation*}
    Since $r > 0$, there exists $\bar{\eta} > 0$ such that $\eta_t \leq \bar{\eta}$ implies $\eta_t \| g^N(\omega_t) - \theta_t \| < r$, thus showing that $(\theta_{t})_+ \in \interior \domain A^*$ and establishing our result.

\end{proof}
\section{PROOFS OF SECTION \ref{sec:convergence}}

\label{app:proof4}

\subsection{Proof of Proposition~\ref{prop:cvgceStochMD}}
\begin{proof}
Consider a current iterate $\omega_t \in \domain A^*$. We are now going to prove a result analogous to \citep[Lemma 4]{dragomir2021}. In order to do so, we compute
\begin{align*}
    d_{A^*}(\omega_*&, \omega_t) - \eta_t d_{f_\pi^D}(\omega_*, \omega_t) - \eta_t (f_\pi^D(\omega_t) - f_\pi^D(\omega_*)) + d_{A^*}(\omega_t, (\omega_{t})_+))\\
    &= d_{A^*}(\omega_*, \omega_t) + \eta_t \langle \nabla f_\pi^D(\omega_t), \omega_* - \omega_t\rangle + d_{A^*}(\omega_t, (\omega_{t})_+))\\
    &= A^*(\omega_*) - A^*((\omega_{t})_+)) - \langle \nabla A^*(\omega_t), \omega_* - \omega_t \rangle - \langle \nabla A^*((\omega_{t})_+)), \omega_t - (\omega_{t})_+) \rangle + \eta_t \langle \nabla f_\pi^D(\omega_t), \omega_* - \omega_t\rangle\\
    &= A^*(\omega_*) - A^*((\omega_{t})_+)) - \langle \nabla A^*(\omega_t), \omega_* - (\omega_{t})_+) \rangle + \eta_t \langle g^N(\omega_t), \omega_t - (\omega_{t})_+) \rangle + \eta_t \langle \nabla f_\pi^D(\omega_t), \omega_* - \omega_t\rangle.
\end{align*}
Now, taking the expectation conditionally on $\omega_t$ in the above, and using Assumption \ref{assumption:estimator}, we obtain
\begin{align*}
    d_{A^*}(\omega_*&, \omega_t) - \eta_t d_{f_\pi^D}(\omega_*, \omega_t) - \eta_t (f_\pi^D(\omega_t) - f_\pi^D(\omega_*)) + \mathbb{E}[d_{A^*}(\omega_t, (\omega_{t})_+))]\\
    &= A^*(\omega_*) - \mathbb{E}[A^*((\omega_{t})_+))] - \mathbb{E}[\langle \nabla A^*(\omega_t), \omega_* - (\omega_{t})_+) \rangle] + \eta_t \mathbb{E}[ \langle \nabla f_\pi^D(\omega_t), \omega_* - (\omega_{t})_+ \rangle]\\
    &= A^*(\omega_*) - \mathbb{E}[A^*((\omega_{t})_+))] - \mathbb{E}[\langle \nabla A^*((\omega_{t})_+)), \omega_* - (\omega_{t})_+\rangle]\\
    &= \mathbb{E}[d_{A^*}(\omega_*, (\omega_{t})_+))].
\end{align*}
The above yields, after rearranging and applying the $m$-strong convexity property of $f^D_\pi$ relatively to $A^*$, the following:
\begin{align*}
    \mathbb{E}[d_{A^*}(\omega_*, (\omega_{t})_+)] - (1 - m\eta_t)d_{A^*}(\omega_*, \omega_t)\leq \eta_t \left(f^D_\pi(\omega_*) - f^D_{\pi,\eta_t,N_t}(\omega_t) \right),
\end{align*}
where we have used the definition of $f^D_{\pi,\eta_t,N_t}$ in Definition \ref{def:variance}.
Since $f^D_\pi(\omega_*) - f^D_{\pi,\eta,N_t}(\omega_t) \leq \eta_t \sigma_{\eta_t,N_t}^2$, we finally obtain that $\mathbb{E}[d_{A^*}(\omega_*, (\omega_{t})_+)] - (1 - m\eta_t)d_{A^*}(\omega_*, \omega_t)\leq \eta_t^2 \sigma_{\eta_t,N_t}^2$. 

We now study the effect of the projection step. We have from Proposition~\ref{prop:BregmanProjection} and \cite[Proposition 3.32 (b)]{bauschke2003} that $\omega_*$ is a fixed point of $\proj_C^{A^*}$. Thus, we have from \cite[Prop.~3.3]{bauschke2003} and \cite[Cor.~3.35]{bauschke2003} that for any $\omega \in \domain A^*$, $d_{A^*}(\omega_*, \proj_C^{A^*}(\omega)) \leq d_{A^*}(\omega_*, \omega)$. Therefore, $d_{A^*}(\omega_*, \omega_{t+1}) \leq d_{A^*}(\omega_*, (\omega_{t})_+)$
from which we deduce that
\begin{equation}
    \label{eq:oneIterateEffect}
    \mathbb{E}\left[ d_{A^*}(\omega_*, \omega_{t+1}) \right] \leq (1-m \eta_t) d_{A^*}(\omega_*, \omega_t) + \eta_t^2 \sigma^2_{\eta_t,N_t}.
\end{equation}

We can obtain the first result with constant sample/batch size $N_t \equiv N$ by iterating  inequality \eqref{eq:oneIterateEffect}.

We now turn to the result with varying sample/batch size and under the assumption $\sigma^2_{\eta,N} \leq \frac{V}{N}$ for some $V > 0$ uniformly in $\eta \in (0,m^{-1}]$. Unrolling the inequality \eqref{eq:oneIterateEffect} with this additional assumption from $t=0$ to $T$ yields
\begin{equation}
    \label{eq:decompositionErrorUpdate}
    \mathbb{E}\left[ d_{A^*}(\omega_{T}, \omega_*) \right] \leq (1 - m \eta)^T \mathbb{E}\left[ d_{A^*}(\omega_{0}, \omega_*) \right] + \eta^2 V \sum_{t=0}^T \frac{(1 - m\eta)^{T-t}}{N_t}.
\end{equation}
We can now perform a change of variable $t^\prime = T - t$, so that the sum is equal to:
\begin{align}
    \sum_{t=0}^T \frac{(1 - m\eta)^{t}}{N_{T-t}} &= \sum_{t=0}^{\thalf} \frac{(1 - m\eta)^{t}}{N_{T-t}} + \sum_{t=\thalf+1}^T \frac{(1 - m\eta)^{t}}{N_{T-t}}\\
    &\leq \sum_{t=0}^{\thalf} \frac{(1 - m\eta)^{t}}{N_{T-\thalf}} + \sum_{t=\thalf+1}^T (1 - m\eta)^{t}.
\end{align}
where we used the fact that $N_t$ is monotonically increasing and $N_0 = 1$. In particular, by plugging $N_t = (t+1)^\gamma$, we obtain:
\begin{align}
    \sum_{t=0}^T \frac{(1 - m\eta)^{t}}{N_{T-t}} &\leq \frac{1}{(T - \thalf + 1)^\gamma}\sum_{t=0}^{\thalf} (1 - m\eta)^{t} + \left((1 - m\eta)^{\thalf+1}\right) \sum_{t=0}^{T - \thalf - 1} (1 - m\eta)^{t}\\
    &\leq \frac{1}{(T - \thalf + 1)^\gamma} \frac{1}{m \eta} + \left((1 - m\eta)^{\frac{T+1}{2}}\right) \frac{1}{m\eta}.
\end{align}
Since $\thalf - 1 < \frac{T}{2} \leq \thalf$, we can finally obtain that $T - \thalf + 1 > \frac{T}{2}$, yielding
\begin{equation}
    \sum_{t=0}^T \frac{(1 - m\eta)^{t}}{N_{T-t}} \leq \frac{1}{m \eta} \left( \frac{2^\gamma}{T^\gamma} + (1-m \eta)^{\frac{T+1}{2}} \right).
\end{equation}
We finally obtain the result by using the above in Eq.~\eqref{eq:decompositionErrorUpdate}. Remark that $T/2$ can actually be replaced by any $\alpha T$ for $0 < \alpha < 1$. 
\end{proof}

\subsection{Proof of Proposition~\ref{prop:cvgceStochMDDecreasing}}
\begin{proof}
    We have from Eq.~\eqref{eq:oneIterateEffect} that at any $t \in \mathbb{N}$ the iterates of Algorithm~\ref{alg} satisfy
    \begin{equation*}
        \mathbb{E}\left[ d_{A^*}(\omega_*, \omega_{t+1}) \right] \leq (1-m \eta_t) d_{A^*}(\omega_*, \omega_t) + \eta_t^2 \sigma^2_{\eta_t,N_t}.
    \end{equation*}
    Taking expectation in the above and dividing by $\eta_t$ yields
    \begin{equation*}
        0 \leq \eta_t \sigma^2_{\eta_t, N_t} + \left(\frac{1}{\eta_t} -m \right) \mathbb{E}\left[ d_{A^*}(\omega_*, \omega_t) \right] - \frac{1}{\eta_t}\mathbb{E}\left[ d_{A^*}(\omega_*, \omega_{t+1}) \right].
    \end{equation*}
    Then, using that $\eta_t = \frac{2}{m(t+2)}$, one obtains
    \begin{equation*}
        0 \leq \frac{2}{m(t+2)}\sigma^2_{\eta_t, N_t} + \frac{m}{2} t \mathbb{E}\left[ d_{A^*}(\omega_*, \omega_t) \right] - \frac{m}{2} (t+2) \mathbb{E}\left[ d_{A^*}(\omega_*, \omega_{t+1}) \right].
    \end{equation*}
    Then, multiplying by $t+1$  yields
    \begin{equation*}
        0 \leq \frac{2(t+1)}{m(t+2)}\sigma^2_{\eta_t, N_t} + \frac{m}{2} t(t+1) \mathbb{E}\left[ d_{A^*}(\omega_*, \omega_t) \right] - \frac{m}{2} (t+1) (t+2) \mathbb{E}\left[ d_{A^*}(\omega_*, \omega_{t+1}) \right].
    \end{equation*}
    Finally, using our additional assumption on the variance, it comes $\frac{t+1}{t+2}\sigma_{\eta_t,N_t}^2 \leq \sigma_{\eta_t,N_t}^2 \leq \frac{V}{N_t}$ and summing  for $t \in \llbracket 0, T \rrbracket$, we obtain that
    \begin{equation*}
        0 \leq \frac{4 V}{m^2} \sum_{t=0}^T \frac{1}{N_t} - (T+1)(T+2) \mathbb{E}\left[ d_{A^*}(\omega_*, \omega_{T+1}) \right]
    \end{equation*}
    from which we obtain the result. 
\end{proof}

\subsection{Proof of Proposition~\ref{prop:LERCConsequences}}

\begin{proof}
    We first state a preliminary result. Recall from Equation \eqref{def:lulu} that for any $\omega \in \domain A^*$, $f_\pi^D(\omega) = A^*(\omega) - \mathbb{E}_{q_\omega}[\log \pi(X)]$. Because $\pi = q_{\theta_\pi} \in \mathcal{Q}_\pi$, $\mathbb{E}_{q_\omega}[\log \pi(X)] = \langle \mathbb{E}_{q_\omega}[\Gamma_\pi(X)], \theta_\pi \rangle - A_\pi(\theta_\pi)$. Using the LERC, we have that $\mathbb{E}_{q_\omega}[\Gamma_\pi(X)] = L_\pi \mathbb{E}_{q_\omega}[\Gamma(X)] = L_\pi \omega$. We thus conclude that 
    \begin{align}
        f_\pi^D(\omega) &= A^*(\omega) - \langle L_\pi^\top \theta_\pi, \omega \rangle + A_\pi(\theta_\pi)\; .
        \label{eq:p6}
    \end{align}

    $(i)$ We first show that our Problem \eqref{eq:VIExpFamConstraint} admits solutions that are necessarily in $\interior \domain A^*$. Remark that solving Problem \eqref{eq:VIExpFamConstraint} is equivalent to minimizing $\omega \longmapsto f_\pi^D(\omega) + \iota_C(\omega)$ where $\iota_C$ is such that
    \begin{equation*}
        \iota_C(\omega) = \begin{cases}
            0 \text{ if } \omega \in C,\\
            +\infty \text{ otherwise}.
        \end{cases}
    \end{equation*}
    Because of \ref{assumption:steep} and Proposition \ref{prop:expFamilyProperties}, and using \eqref{eq:p6}, $f_\pi^D$ is lower semi-continuous. Because of \ref{assumption:constraint}, $\iota_C$ is also lower semi-continuous. Therefore, $f_\pi^D + \iota_C$ is lower semi-continuous. Further, $f_\pi^D + \iota_C$ is coercive if $C$ is compact ($\mathcal{H}$ being finite-dimensional, this implies $C$ is bounded) or if $L_\pi^\top \theta_\pi \in \interior \domain A$, from \cite[Fact 2.11]{bauschke1997}. Therefore, $f_\pi^D + \iota_C$ is lower semi-continuous and coercive, so there exist minimizers. 
    
    Now, let us show that these solutions necessarily belong to $\interior \domain A^*$. Assumption \ref{assumption:constraint} implies that $C \cap \interior \domain f_\pi^D \neq \emptyset$ because of the LERC and \eqref{eq:p6}. So we can apply \citep[Proposition 6.19 (vii)]{bauschke2011} and \citep[Proposition 26.5 (a)]{bauschke2011} to have that at any minimizer $\omega_*$, $\partial A^*(\omega_*) \neq \emptyset$, where $\partial A^*$ is the subdifferential of $A^*$ \citep[Chapter 16]{bauschke2011}. This necessarily implies that $\omega_* \in \interior \domain A^*$, because due to $A^*$ being Legendre (from Lemma \ref{lemma:propertiesA} and Proposition \ref{prop:legendreProperties}), $\domain \partial A^* = \interior \domain A^*$ from \cite[Theorem 26.1]{rockafellar1970}. Thus, we have shown that solutions of Problem \eqref{eq:VIExpFamConstraint} exist and are in $\interior \domain A^*$.

    Finally, let us prove the uniqueness. Because of the constraint and previous points, possible solutions are in $C \cap \interior \domain A^*$. Because of Assumption \ref{assumption:steep} and Proposition \ref{prop:expFamilyProperties}, $A^*$ is strictly convex on $\interior \domain A^*$, so $f_\pi^D$ is also strictly convex on $\interior \domain A^*$ due to the LERC and \eqref{eq:p6}. Additionally, Assumption \ref{assumption:constraint} implies that $\iota_C$ is convex on $\mathcal{H}$. Thus, $f_\pi^D + \iota_C$ is strictly convex on $C \cap \interior \domain A^*$, showing that there exist only a unique solution $\omega_*$ to Problem \eqref{eq:VIExpFamConstraint}, and establishing the result.

    We now show that if $L_\pi^\top \theta_\pi \in \interior \domain A$, then $\omega_* = \proj_C^{A^*}(\nabla A(L_\pi^\top \theta_\pi))$. Because of \eqref{eq:p6}, the fact that $A^*$ is Legendre, and that $\omega_*$ minimizes $\omega \longmapsto f_\pi^D(\omega) + \iota_C(\omega)$, $\omega_*$ satisfies the optimality conditions
    \begin{equation*}
        0 \in \nabla A^*(\omega_*) - \nabla A^*(\nabla A(L_\pi^\top \theta_\pi)) + \partial \iota_C(\omega_*),
    \end{equation*}
    where $\partial \iota_C$ is the subdifferential of $\iota_C$ \citep[Chapter 16]{bauschke2016}. We can thus recognize that $\omega_*$ satisfies the optimality conditions associated to the minimization of $\omega \longmapsto d_{A^*}(\omega, \nabla A(L_\pi^\top \theta_\pi)) + \iota_C(\omega)$, which is the optimization problem associated with the computation of $\proj_C^{A^*}(\nabla A(L_\pi^\top \theta_\pi))$. Because of Proposition \ref{prop:BregmanProjection}, this problem admits a unique solution that satisfies the optimality conditions, so we can deduce that $\omega_* = \proj_C^{A^*}(\nabla A(L_\pi^\top \theta_\pi))$, showing our result.

    $(ii)$ We can use Equation \eqref{eq:p6} above so that
    $f^D_\pi$ is the sum of $A^*$ and an affine term in $\omega$. Since the Bregman divergence induced by an affine function is zero and the Bregman divergence induced by a sum of functions is the sum of the Bregman divergences induced by each function, we obtain that $d_{f^D_\pi} = d_{A^*}$, thus showing the result.
\end{proof}

\subsection{Proof of Proposition~\ref{prop:varianceBonnetPrice}}
\begin{proof}
The target distribution $\pi$ is here assumed  to be a  Gaussian distribution with mean $\mu_\pi$ and covariance $\Sigma_\pi$. 
As ${\cal Q}$ is the Gaussian family, then
     $\pi$ satisfies the LERC  with respect to ${\cal Q}$, see Definition~\ref{def:linearEnlargement}.
     It follows that $f^D_\pi$ is $1$-smooth relatively to $A^*$ due to Proposition~\ref{prop:LERCConsequences} $(ii)$. We thus have from Proposition~\ref{prop:varianceBoundHadrien}  and Proposition~\ref{prop:expFamilyProperties} (iv) that
    \begin{equation*}
        \sigma_{\eta,N}^2 \leq \frac{1}{\eta^2} \sup_{\omega \in C \cap \interior \domain A^*} \left\{ \mathbb{E}\left[ d_A(\theta_+, \Bar{\theta}_+) \right] \right\},
    \end{equation*}
    where  $\theta_+ = (1-\eta) \theta + \eta g^N(\omega)$ and   $\Bar{\theta}_+ = (1-\eta) \theta + \eta \nabla \mathbb{E}_{q_\omega}[\log \pi(X)]$, so that  $\mathbb{E}[\theta_+]= \Bar{\theta}_+$. 

Let us  first specify these latter expressions.
When $q_\omega$ is a Gaussian distribution, as already used in Section \ref{subsec:BonnetPrice}, we have more generally with some arbitrary differentiable function $l$ that
    \begin{align*}
        \nabla_{\omega_1} \mathbb{E}_{q_{\omega}}[l(X)] &= \mathbb{E}_{q_\omega}[\nabla l(X)] - \mathbb{E}_{q_\omega}[\nabla^2 l(X)]\mu\\
        \nabla_{\omega_2} \mathbb{E}_{q_\omega}[l(X)] &= \frac{1}{2} \mathbb{E}_{q_\omega}[\nabla^2 l(X)],
    \end{align*}
    which can be estimated with the following estimators
    \begin{align*}
        g^N(\omega)_1 &= \frac{1}{N} \sum_{n=1}^N \left( \nabla l(X_n) - \nabla^2 l(X_n) \mu \right)\\
        g^N(\omega)_2 &= \frac{1}{N} \sum_{n=1}^N \frac{1}{2} \nabla^2 l(X_n).
    \end{align*}
    These estimators can be checked to be unbiased, thus satisfying Assumption \ref{assumption:estimator}.

    Thus,  for $\log \pi(x) =  - \frac{1}{2}(x- \mu_\pi)^{\top} \Sigma_\pi^{-1}(x - \mu_\pi) + \text{constant}$, it comes
    \begin{align*}
        \nabla_{\omega_1} \mathbb{E}_{q_{\omega}}[\log \pi(X)] &= - \Sigma_\pi^{-1} (\mu - \mu_\pi) + \Sigma_\pi^{-1}\mu = \Sigma_\pi^{-1} \mu_\pi &
        g^N(\omega)_1 &= \Sigma_\pi^{-1} \left( \mu_\pi + \mu - \frac{1}{N} \sum_{n=1}^N X_n \right)\\
        \nabla_{\omega_2} \mathbb{E}_{q_\omega}[\log \pi(X)] &= -\frac{1}{2} \Sigma_\pi^{-1} &
        g^N(\omega)_2 &= -\frac{1}{2} \Sigma_\pi^{-1}.
    \end{align*}
    We can readily notice that $g^N(\omega) \in \interior \domain A$ almost surely, thus establishing \ref{assumption:wellposedness} using Proposition \ref{prop:sufficientWellPosedness}.
    
We can then also remark that $\Bar{\theta}_+ = (1-\eta) \theta + \eta \theta_\pi$ and denoting $\bar{\theta}_+ = (\bar{\theta}_{+,1}, \bar{\theta}_{+,2})$, $\theta_+ = (\theta_{+,1}, \theta_{+,2})$ that  $\bar{\theta}_{+,2} = \theta_{+,2}$.

    Going back to computing $d_A(\theta_+, \Bar{\theta}_+)$ in the above bound, it is equivalent to computing the KL divergence between 2 Gaussian distributions respectively defined by $\Bar{\theta}_+$ and $\theta_+$ (Proposition~\ref{prop:expFamilyProperties} (iv)). The expression of this KL is
    $$ KL(q_{\Bar{\theta}_+},q_{\theta_+}) = A(\theta_+)- A(\Bar{\theta}_+) - \langle\theta_+ -\Bar{\theta}_+,\Bar{\omega}_+\rangle\; .$$
    
    For  a Gaussian distribution with natural parameter $\theta = (\theta_1, \theta_2)$ the expression of the log-partition function is $A(\theta) = -  \theta_1^\top (4 \theta_2)^{-1} \theta_1 - \frac{1}{2} \logdet(-2\theta_2)$. It follows that using 
    $\bar{\theta}_{+,2} = \theta_{+,2}$,
$$A(\theta_+)- A(\Bar{\theta}_+) = \theta_{+,1}^\top \left(-4 \theta_{+,2} \right)^{-1} \theta_{+,1} - \bar{\theta}_{+,1}^\top \left(-4 \theta_{+,2} \right)^{-1} \bar{\theta}_{+,1}\;,$$
and
$$ \langle\theta_+ -\Bar{\theta}_+,\Bar{\omega}_+\rangle = \langle\theta_{+,1} -\Bar{\theta}_{+,1},\Bar{\omega}_{+,1}\rangle $$
where $\Bar{\omega}_{+,1}= \Bar{\mu}_+ =  \left(-2 \theta_{+,2} \right)^{-1} \bar{\theta}_{+,1}$.

    Combining both terms, it comes 
    \begin{equation*}
        d_A(\theta_+, \bar{\theta}_+) =  \theta_{+,1}^\top \left(-4 \theta_{+,2} \right)^{-1} \theta_{+,1} - \bar{\theta}_{+,1}^\top \left(-4 \theta_{+,2} \right)^{-1} \bar{\theta}_{+,1} - \langle 2 \left(-4 \theta_{+,2} \right)^{-1} \bar{\theta}_{+,1}, \theta_{1,+}- \bar{\theta}_{+,1}\rangle.
    \end{equation*}
    From here, we deduce, using that $\mathbb{E}[\theta_{1,+}]=\bar{\theta}_{+,1}$, 
    \begin{align*}
        \mathbb{E}\left[d_A(\theta_+, \bar{\theta}_+) \right] &= \mathbb{E}\left[ \theta_{+,1}^\top \left(-4 \theta_{+,2} \right)^{-1} \theta_{+,1} \right] - \mathbb{E}\left[ \theta_{+,1} \right]^\top \left(-4 \theta_{+,2} \right)^{-1} \mathbb{E}\left[ \theta_{+,1} \right]\\
        &= \mathbb{E}\left[ (\theta_{+,1} - \mathbb{E}\left[ \theta_{+,1} \right])^\top \left(-4 \theta_{+,2} \right)^{-1} (\theta_{+,1} - \mathbb{E}\left[ \theta_{+,1} \right])  \right]\\
        &= \eta^2 \mathbb{E}\left[ (g^N(\omega)_1 - \theta_{\pi,1})^\top \left(-4 \theta_{+,2} \right)^{-1} (g^N(\omega)_1 - \theta_{\pi,1})  \right]\\
        &= \eta^2 \mathbb{E}\left[ (\mu - \frac{1}{N}\sum_{n=1}^N X_n)^\top \Sigma_\pi^{-1}\left(-4 \theta_{+,2} \right)^{-1} \Sigma_\pi^{-1} (\mu - \frac{1}{N}\sum_{n=1}^N X_n)  \right]\\
        &= \frac{\eta^2}{N^2} \left\langle \Sigma_\pi^{-1}\left(-4 \theta_{+,2} \right)^{-1} \Sigma_\pi^{-1}, \mathbb{E}\left[ (\sum_{n=1}^N (X_n-\mu))(\sum_{n=1}^N (X_n-\mu))^\top \right] \right\rangle.
    \end{align*}
    Since $X_n - \mu \sim \mathcal{N}(0,\Sigma)$ for any $n \in \llbracket 1,N \rrbracket$ the matrix-valued random variable $(\sum_{n=1}^N (X_n-\mu))(\sum_{n=1}^N (X_n-\mu))^\top$ follows a Wishart distribution with $N$ degrees of freedom and scale matrix $\Sigma$, so its expectation is $N \Sigma$. Therefore, we have shown that
    \begin{equation*}
        \mathbb{E}\left[d_A(\theta_+, \bar{\theta}_+) \right] = \frac{\eta^2}{N} \left\langle \Sigma_\pi^{-1}\left(-4 \theta_{+,2} \right)^{-1} \Sigma_\pi^{-1}, \Sigma \right\rangle.
    \end{equation*}

    We now have to bound the above quantity on $\interior \domain A^*$. The trace formulation of the previous identity is: 
    
    \begin{equation*}
        \mathbb{E}\left[d_A(\theta_+, \bar{\theta}_+) \right] = \frac{\eta^2}{N} \tr \left( \Sigma_\pi^{-1}\left(-4 \theta_{+,2} \right)^{-1} \Sigma_\pi^{-1} \Sigma\right).
    \end{equation*}
    We now handle the term in $\theta_{+,2}$. By definition $-4 \theta_{+,2}=2(1-\eta)\Sigma^{-1} + 2\eta \Sigma_\pi^{-1}$. We introduce the map
    \begin{align*}
        g(\Sigma, \eta) &=  \frac{1}{2}\tr \left( \Sigma_\pi^{-1}\left((1-\eta) \Sigma^{-1} + \eta \Sigma_\pi^{-1} \right)^{-1} \Sigma_\pi^{-1}\Sigma\right)\\
        &= \frac{1}{2}\tr \left( \tilde{\Sigma}\left((1-\eta) I + \eta \tilde{\Sigma} \right)^{-1} \tilde{\Sigma}\right),
    \end{align*}
    with $ \tilde{\Sigma} = \Sigma^\frac{1}{2} \Sigma_\pi^{-1} \Sigma^\frac{1}{2}$. Expressing $\tilde{\Sigma} = Q^\top \Delta Q$, where $Q^\top Q = Q Q^\top = I$ and $\Delta$ is a diagonal matrix, we obtain: 
    \begin{align*}
        g(\Sigma, \eta) &= \frac{1}{2}\tr \left( Q^\top \Delta Q \left((1-\eta) Q^\top Q + \eta Q^\top \Delta Q \right)^{-1} Q^\top \Delta Q\right)\\
        &=  \frac{1}{2}\tr \left(Q^\top Q \Delta Q\left(Q^\top [(1-\eta) I + \eta \Delta] V \right)^{-1} Q^\top \Delta \right)\\
        &= \frac{1}{2}\tr \left( \Delta \left[(1-\eta) I + \eta \Delta\right]^{-1} \Delta\right)\\
        &=  \frac{1}{2} \sum_{i=1}^d \frac{\Delta_{ii}^2}{1 - \eta + \eta \Delta_{ii}}.
    \end{align*}
    One can directly notice that $g(\Sigma, \eta)$ remains bounded for any $\eta \in (0,1]$ as any of the $\Delta_{ii}$ goes to $0$. However, we need the $\Delta_{ii}$ to remain bounded to ensure boundedness of $g(\Sigma, \eta)$. While this is not true in general, we actually do not need to take the supremum over $ C \cap \interior \domain A^*$, but actually only need to take the value $g(\Sigma, \eta)$ at the point that achieves the supremum in Definition~\ref{def:variance} (which is achieved). While possible, this derivation is rather cumbersome. 

    Instead, suppose that $C = \mathcal{H}$, i.e., no constraints are applied (the constrained case will be treated later on). Then, we notice that at each step, the updates in $\Sigma$ are such that the updated values of $\Sigma$ will belong to the set $C' = \{ \Sigma \in \mathbb{S}^d,\, \Sigma = ((1-\eta)\Sigma_\pi^{-1} + \eta \Sigma_0)^{-1},\,\eta \in [0,1]\}$. Because $\Sigma_0,\Sigma_\pi \in \mathbb{S}^d_{>0}$ by assumption, $C' \subset \mathbb{S}^d_{>0}$ and is compact.
 
    The above ensures that the eigenvalues of $\Sigma$ are bounded from above, so that the $\Delta_{ii}$ are also bounded  from above. This implies that there exists $\widetilde{V} > 0$ such that
    \begin{equation*}
        \mathbb{E}\left[d_A(\theta_+, \bar{\theta}_+) \right] \leq \frac{\eta^2\widetilde{V}}{N},\,\forall \omega \in \interior \domain A^*.
    \end{equation*}
    We can therefore apply Proposition~\ref{prop:varianceBoundHadrien} to obtain the result. 

    In the case where $C \neq \mathcal{H}$, meaning that additional constraints are imposed, one can simply notice that we take a supremum on $C \cap \interior \domain A^*$ in the definition of $\sigma^2_{\eta, N}$ instead of a supremum on $\interior \domain A^*$, which is a bigger set. Therefore, the bound that we derived for $C = \mathcal{H}$ also holds when $C$ is smaller.
    
\end{proof}

\subsection{Proof of Proposition~\ref{prop:boundVarianceSubsampling}}
\begin{proof}
    We can readily check that $g^N$ is an unbiased estimator of $\nabla_\omega \mathbb{E}_{q_\omega}[\log \pi(X)]$ and that it belongs almost surely to $\interior \domain A^*$, allowing to apply Proposition \ref{prop:sufficientWellPosedness}. Thus, Assumptions \ref{assumption:wellposedness} and \ref{assumption:estimator} are satisfied.
    
    We are going to bound $\sigma_{\eta,N}^2$ using Proposition~\ref{prop:varianceBoundHadrien} and the fact that $f^D_\pi$ is $1$-smooth relatively to $A^*$ from Proposition~\ref{prop:LERCConsequences} $(ii)$. Proposition~\ref{prop:varianceBoundHadrien} leads to 
    \begin{align*}
        \sigma_{\eta,N}^2 &\leq \frac{1}{\eta^2} \sup_{\omega \in \domain A^*} \left\{ \mathbb{E} \left[ d_{A^*}(\overline{\omega}_+, \omega_+) \right] \right\}\\
        &\leq \frac{1}{\eta^2} \sup_{\omega \in \domain A^*} \left\{ \mathbb{E} \left[ d_{A^*}(\overline{\omega}_+, \omega_+) + d_{A^*}(\omega_+, \overline{\omega}_+) \right] \right\}
        \end{align*}
        The right-hand side above is then equal to
        \begin{align*}
        \frac{1}{\eta^2} \sup_{\omega \in \domain A^*} \left\{ \mathbb{E} \left[ \langle \nabla A^*(\overline{\omega}_+) - \nabla A^*(\omega_+), \overline{\omega}_+ - \omega_+ \rangle \right] \right\}
        &= \frac{1}{\eta} \sup_{\omega \in \domain A^*} \left\{ \mathbb{E} \left[ \langle g^N(\omega) - \nabla f^D_\pi(\omega), \overline{\omega}_+ - \omega_+ \rangle \right] \right\}
    \end{align*}
    so that eventually
    \begin{align*}
        \sigma_{\eta,N}^2 &\leq 
        \frac{1}{\eta} \sup_{\omega \in \domain A^*} \left\{ \mathbb{E} \left[ \langle g^N(\omega) - \nabla f^D_\pi(\omega), \overline{\omega}_+ - \omega_+ \rangle \right] \right\}
    \end{align*}
    
    This latter right-hand side expression has been proposed by \cite{hanzely2021} as a variance-like quantity to analyse mirror descent. In the case of NGVI for Bayesian linear regression with Gaussian prior, likelihood, and approximating family, this quantity has been upper-bounded in \cite[Lemma 4]{wu2024} by the constant
    \begin{equation*}
        V_2 = (s s_1 + \frac{1}{2} s^2 s_2 + 2 s^2 b \sqrt{s_1 s_2} M + s^3 b^2 s_2 M^2) \frac{M^2}{\sigma^4 N},
    \end{equation*}
    where $s,s_1,s_2,b$ are given by the following expressions,
    \begin{align*}
        s &= \max \{ 1, \| \Sigma_0 \| \}, & b &= \max_{1 \leq n \leq N} \| y_m z_m \|,\\
        s_1 &= \mathbb{E}_{\mathcal{U}_M} \left[ \left\| y_X z_X - \frac{1}{N} \sum_{m=1}^M y_m z_m \right\|^2 \right], & s_2 &= \mathbb{E}_{\mathcal{U}_M} \left[ \left\| z_X z_X^\top - \frac{1}{N} \sum_{m=1}^M z_m z_m^\top \right\|^2 \right].
    \end{align*}

\end{proof}

\section{PROJECTION OPERATORS}
\label{app:projections}

In Section \ref{sec:presentation}, we have introduced the possibility of further restricting the search space from distributions $q_\theta \in \mathcal{Q}$ with $\theta \in \interior \domain A^*$ to distributions $q_\theta \in \mathcal{Q}$ with $\theta \in C \cap \interior \domain A^*$. This can be used to safeguard iterates of Algorithm \ref{alg} from problematic behaviours (e.g., covariance matrices becoming singular) that may be due, for instance, to a number of samples that is too low. Another interest of this constraint is to incorporate information about the target distribution $\pi$. This constraint translates to an additional projection step in Algorithm \ref{alg}. We now give two situations where this projection operator admits a closed form.

\subsection{Gaussian distributions with constrained covariance matrices}
\label{subsec:projectionBoundedEigen}

In the case of Gaussian distributions, an example of a constraints space $C$ with an explicit projection operator is defined for $0 < \alpha < \beta$ as:
\begin{equation}
    \label{eq:boundedEigenvalSet}
    C = \{ (\mu, \Sigma + \mu \mu^{\top}) \in \mathcal{H} \text{ s.t. } \alpha I \preccurlyeq \Sigma \preccurlyeq \beta I\}.
\end{equation}

\begin{proposition}
    \label{prop:projectionEigenvalue}
    Consider $\mathcal{Q}$ the exponential family of Gaussian distributions, the set $C$ defined in \eqref{eq:boundedEigenvalSet}, and $\omega = (\mu, \Sigma + \mu \mu^\top)$. Then, $\proj_C^{A^*}(\omega) = (\mu, \Sigma_P + \mu \mu^\top)$ with $\Sigma_P$ being a version of $\Sigma$ whose eigenvalues larger (resp.~smaller) than $\beta$ (resp.~$\alpha$) have been replaced by $\beta$ (resp.~$\alpha$).
\end{proposition}

\begin{proof}
    We recall that $\proj_C^{A^*}(\omega) = \argmin_{\omega'} \left\{ \iota_C(\omega') + d_{A^*}(\omega', \omega) \right\}$. Since $\omega \in C$ depends only on $\Sigma$, we rewrite the optimisation problem associated to the projection operator in terms of the mean and covariance matrices, where $\omega= (\mu, \Sigma + \mu \mu^\top)$. We then get that $(\mu_P, \Sigma_P + \mu_P \mu_P^\top) = \proj_C^{A^*}(\omega)$ are the solutions to the optimisation problem
    \begin{equation*}
        \min_{\mu' \in \mathbb{R}^d, \Sigma' \in \mathbb{S}^d} \left\{ \iota_{\widetilde{C}}(\Sigma') + \frac{1}{2} \left(- \logdet \Sigma' + (\mu - \mu')^\top \Sigma^{-1}(\mu - \mu') + \langle \Sigma^{-1}, \Sigma \rangle_{\mathbb{S}^d} \right) \right\}.
    \end{equation*}
    We thus observe that $\mu_P = \mu$ and that $\Sigma_P$ is a solution to
    \begin{equation*}
        \min_{\Sigma' \in \mathbb{S}^d} \left \{ \iota_{\widetilde{C}}(\Sigma') - \frac{1}{2} \logdet \Sigma' + \langle \Sigma^{-1}, \Sigma' \rangle_{\mathbb{S}^d} \right \}.
    \end{equation*}
    We will now show using \cite[Theorem 1]{benfenati2020} that this problem can be reduced to a problem on $\mathbb{R}^d$ involving only the eigenvalues of $\Sigma'$.
    
    If we define by $\lambda' \in \mathbb{R}^d$ the set of eigenvalues of $\Sigma'$, we can remark that $\iota_{\widetilde{C}}(\Sigma') = \sum_{i=1}^d \iota_{[\alpha, \beta]}(\lambda'_i)$ and $-\frac{1}{2} \logdet \Sigma' = -\frac{1}{2} \sum_{i=1}^d \log \lambda_i$, meaning that these two functions only depend on the eigenvalues of $\Sigma'$. Further, each of these functions is lower semicontinuous, their sum is coercive, and $\Sigma$ is such that there exists an orthonormal matrix $Q$ and vector $\delta \in \mathbb{R}^d$ such that $\Sigma^{-1} = Q\diag(\delta) Q^{\top}$, meaning that we can invoke \cite[Theorem 1]{benfenati2020} to show that $\Sigma_P = Q \diag(\lambda_P) Q^\top$ with $\lambda_P$ being a solution to
    \begin{equation*}
        \min_{\lambda' \in \mathbb{R}^d} \left\{ \sum_{i=1}^d \iota_{[\alpha, \beta]}(\lambda'_i) - \frac{1}{2} \sum_{i=1}^d \log \lambda'_i + \frac{1}{2}\sum_{i=1}^d \lambda_i \delta_i \right\}.
    \end{equation*}
    This problem has a separable structure, implying that for each $i \in \llbracket 1,d \rrbracket$,
    \begin{equation*}
        (\lambda_P)_i = \argmin_{\alpha \leq \lambda' \leq \beta} \left\{ -\log \lambda' + \delta_i \lambda' \right\}
    \end{equation*}
    Since for each $i \in \llbracket 1,d \rrbracket$, $\delta_i = \frac{1}{\lambda_i}$ where $\lambda \in \mathbb{R}^d_{>0}$ is the vector of eigenvalues of $\Sigma$, we can finally check the result by solving the above problem using the Karush-Kuhn-Tucker conditions.
\end{proof}

\subsection{Diagonal Gaussians with non-negative means}
\label{subsec:projectionNonNegMean}

In this section, we show that when $\mathcal{Q}$ is the family of Gaussians with diagonal covariance, introduced in Section \ref{subsec:diagonalGaussian}, the projection to
\begin{equation}
    \label{eq:nonNegativityConstraint}
    C = \{(\mu, \sigma^2 + \mu \bullet \mu) \in \mathbb{R}^d \times \mathbb{R}^d \text{ s.t. } \mu \in \mathbb{R}^d_{\geq 0}\}
\end{equation}
admits an explicit expression.

\begin{proposition}
    Consider the exponential family of Gaussians with diagonal covariance, the set $C$ defined in \eqref{eq:nonNegativityConstraint}, and $\omega = (\mu, \sigma^2 + \mu \bullet \mu)$. Then, $\proj_C^{A^*}(\omega) = (\mu_P, \sigma^2 + \mu_P \bullet \mu_P)$ with $(\mu_P)_i = \max(0, \mu_i)$ for $i \in \llbracket 1,d \rrbracket$.
\end{proposition}

\begin{proof}
    The projection of $\omega \in \interior \domain A^*$ on $C$ as defined in \eqref{eq:nonNegativityConstraint} consists in solving
    \begin{equation*}
        \omega_P = \argmin_{\omega' \in \interior \domain A^*} \left( g(\omega') + d_{A^*}(\omega', \omega) \right),
    \end{equation*}
    where $g(\omega) = \iota_{\mathbb{R}_{\geq 0}^d}(\omega_1)$. Therefore, $\omega_P$ solves the optimality conditions $0 \in \nabla A^*(\omega_P) - \nabla A^*(\omega) + \partial g(\omega)$ where $\partial g$ is the subdifferential of $g$ (see \cite[Chapter 16]{bauschke2011}). Since $g$ depends only on $(\omega_P)_1 = \mu_P$ and is separable, we see that $\omega_P$ solves
    \begin{align*}
        0 &\in \frac{1}{(\sigma^2_P)_i}(\mu_P)_i - \frac{1}{(\sigma^2)_i}(\mu)_i + \partial \iota_{\mathbb{R}_{\geq 0}}((\mu_P)_i)\\ 
        0 &= \frac{1}{(\sigma^2_P)_i} - \frac{1}{(\sigma^2)_i},
    \end{align*}
    for any $i \in \llbracket 1,d \rrbracket$. We thus readily obtain that $\sigma^2_P = \sigma^2$. Now, using that fact that $\partial \iota_{\mathbb{R}_{\geq 0}}(s)$ is the normal cone $N_{\mathbb{R}_{\geq 0}}(s)$ for an scalar $s$ (see \cite[Example 16.12]{bauschke2011}), and that 
    \begin{equation*}
        N_{\mathbb{R}_{\geq 0}}(s) = \begin{cases}
            \mathbb{R}_{\leq 0} \text{ if } s=0,\\
            \emptyset \text{ if } s<0,\\
            \{0\} \text{ if } s>0,
        \end{cases}
    \end{equation*}
    we obtain that if $\mu_i <0$, $(\mu_P)_i = 0$ and that $\mu_i = (\mu_P)_i$ else.
\end{proof}

\section{ADDITIONAL NUMERICAL ILLUSTRATIONS}
\label{app:additionalExperiments}

\subsection{Additional details about the experiments of Section \ref{sec:experiments}}

For a numerical illustration of our results in the case of data subsampling estimators, we consider the \cite{dataset} dataset \citep{kaya2019}. We use the NOx emissions as a response variable $Y$ and the other variables (excluding the CO variable, thus $d=9$) as fixed covariates.

In the case of experiments performed in the conjugate setting, all the runs are initialized with initial mean simulated uniformly in $[-5,5]^d$ and initial covariance matrix being equal $10 I$. In the case of the logistic regression task, the initialization is similar except that the initial covariance matrix is equal to $0.5 I$.

We recall that the source code can be found at \url{https://github.com/tGuilmeau/Projected_Stochastic_NGVI/}.

\subsection{Bonnet and Price estimators with a Gaussian target}

In this Section, we consider the same setting as in Section \ref{sec:experiments}, namely a Gaussian target in dimension $d=10$ with Gaussian approximating distributions, and Algorithm \ref{alg} being run with the Bonnet and Price gradient estimators without performing a projection step ($C = \mathbb{R}^d \times \mathbb{S}^d$). In this setting, we investigate the influence of the hyperparameters of Algorithm \ref{alg}.

\paragraph{Impact of $\eta$ and $N$ in the fixed step sizes and fixed sample sizes schedule} When $\eta_t \equiv \eta$ and $N_t \equiv N$, Proposition \ref{prop:cvgceStochMD} predicts a geometric convergence, with rate $(1-\eta)$, of the iterates of Algorithm \ref{alg} to a neighborhood of the optimum $\omega_*$, whose size is controlled by $\sigma_{\eta, N}^2$, defined in Definition \ref{def:variance}. We further experimentally investigate the effect of $\eta$ and $N$.

Figure \ref{fig:BonnetPriceIterationsDifferentEta} shows the performance of Algorithm \ref{alg} across different values of the step size $\eta$. Figure \ref{fig:BonnetPriceIterationsDifferentEta} shows that low values of $\eta$ yield a slower convergence, but to a better value. Indeed, the slower convergence comes from the fact that the geometrically decreasing term in Proposition \ref{prop:cvgceStochMD} decreases with $(1-\eta)^t$. The fact that the iterates are closer to $\omega_*$ when $\eta$ is low indicates that $\sigma^2_{\eta,N}$ decreases with the step size $\eta$. 

Figure \ref{fig:BonnetPriceIterationsDifferentN} shows the performance of Algorithm \ref{alg} when different values of $N$ are considered, keeping $\eta$ fixed. All tested values of $N$ lead to a geometric decrease with the same rate, since all runs share the same step size $\eta$, but the size of the neighborhood around the solution $\omega_*$ decreases as $N$ increases. This can be expected from Proposition \ref{prop:varianceBonnetPrice}, which exactly predicts that $\sigma^2_{\eta, N} \leq \frac{V}{N}$ for some $V > 0$, uniformly in $\eta \in (0,1]$.

\begin{figure}[H]
    \centering
    \begin{subfigure}[t]{.49\linewidth}
        \centering
        \includegraphics[width=\linewidth]{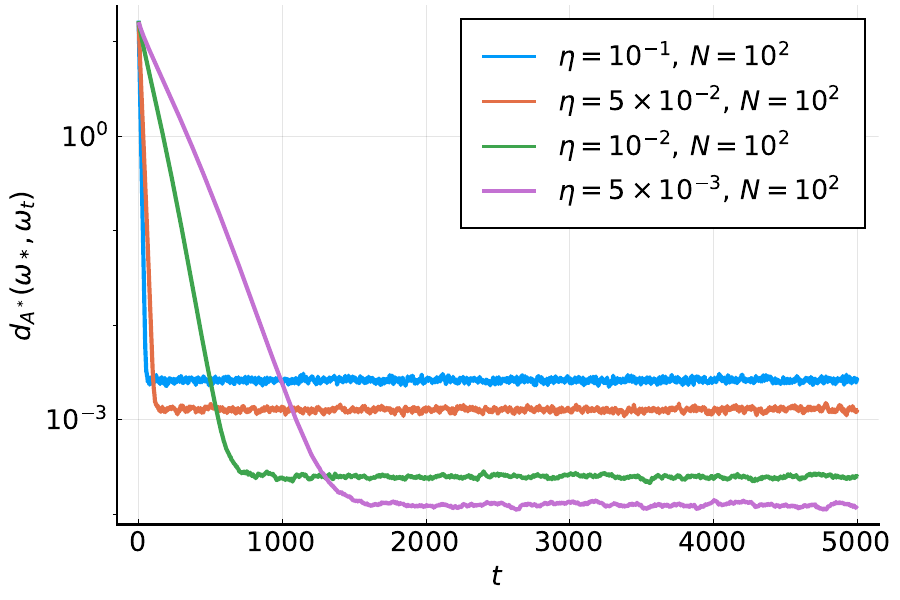}
        \caption{$N=100$ and different values of $\eta$ are tested}
        \label{fig:BonnetPriceIterationsDifferentEta}
    \end{subfigure}\hfill
    \begin{subfigure}[t]{.49\linewidth}
        \centering
        \includegraphics[width=\linewidth]{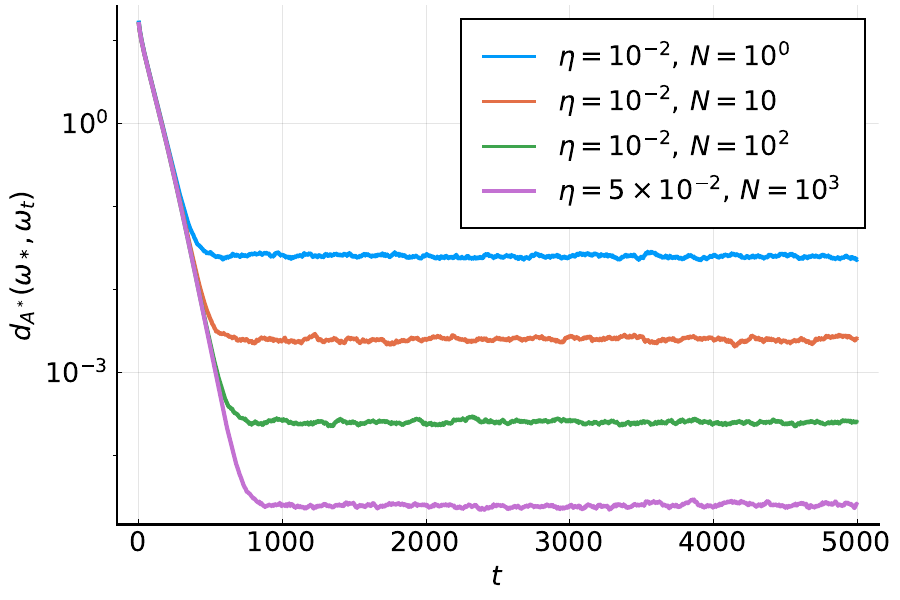}
        \caption{$\eta = 10^{-2}$ and different values of $N$ are tested}
        \label{fig:BonnetPriceIterationsDifferentN}
    \end{subfigure}
    
    \caption{Mean Bregman divergence between current and optimal parameters, over $100$ runs, for NGVI schedules with constant $\eta$ and constant sample size $N$}
    \label{fig:BonnetPriceNeighborhoodSize}
\end{figure}

\paragraph{Influence of the parameter $\gamma$ in the choice of the sample size} In Proposition \ref{prop:cvgceStochMD}, we showed that when constant step sizes $\eta_t \equiv \eta$ are chosen with increasing sample sizes $N_t = (t+1)^\gamma$, then, the iterates converge to the minimizer at a $\mathcal{O}(\frac{1}{T^\gamma})$ rate. We investigate this effect in the following.

Figure \ref{fig:BonnetPriceIterationsDifferentGamma} shows that, in terms of iterations, a higher value of $\gamma$ yields iterates that get closer to the minimizer at a faster rate, as predicted by Proposition \ref{prop:cvgceStochMD}. Note that in the first iterations, all the choice of $\gamma$ lead to the same geometric rate, the difference appearing later.

Figure \ref{fig:BonnetPriceComputationDifferentGamma} highlights the compromise that needs to be made in terms of computational budget when selecting $\gamma$. Indeed, if a larget value of $\gamma$ increase the performance in terms of iteration count, the resulting algorithm is more expensive in terms of computational budget. Figure \ref{fig:BonnetPriceComputationDifferentGamma} suggests that there is a "optimal" value $\gamma$ such that, for a given computational budget, allows to obtain the best performance.

\begin{figure}[H]
    \centering
    \begin{subfigure}[t]{.49\linewidth}
        \centering
        \includegraphics[width=\linewidth]{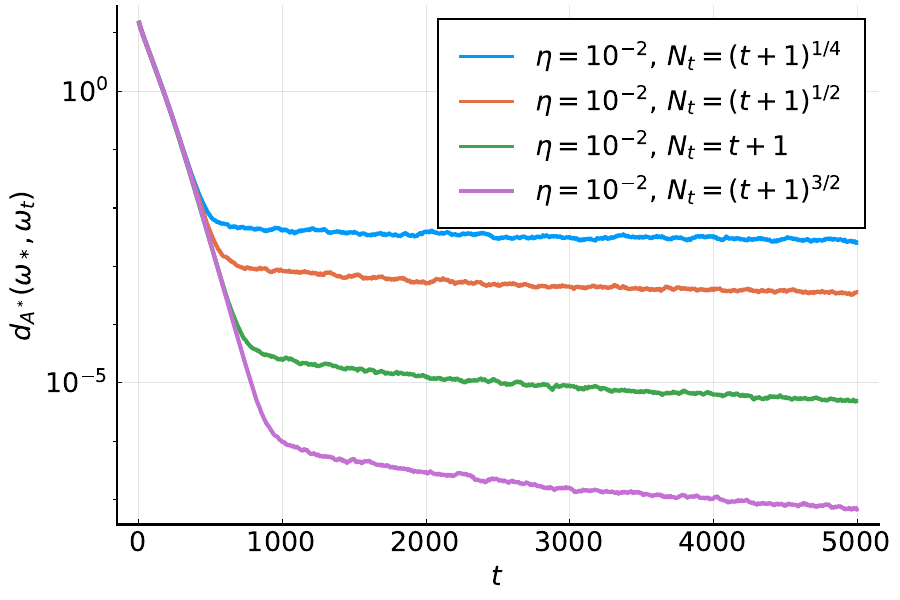}
        \caption{wrt iterations}
        \label{fig:BonnetPriceIterationsDifferentGamma}
    \end{subfigure}\hfill
    \begin{subfigure}[t]{.49\linewidth}
        \centering
        \includegraphics[width=\linewidth]{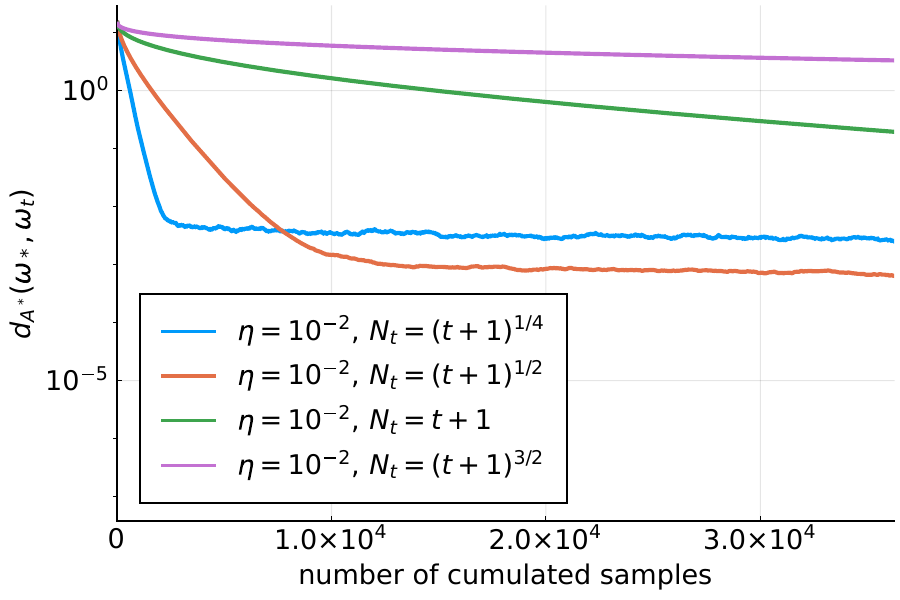}
        \caption{wrt number of samples}
        \label{fig:BonnetPriceComputationDifferentGamma}
    \end{subfigure}
    
    \caption{Mean Bregman divergence between current and optimal parameters, over $100$ runs, for NGVI schedules with constant $\eta$ and increasing sample size $N_t = (t+1)^\gamma$}
    \label{fig:BonnetPriceInfluenceGamma}
\end{figure}

\subsection{Projection impact in a robust regression setting}

In this section, we illustrate the positive effect of a projection step in Algorithm~\ref{alg} in a non-recoverable case, which is not theoretically covered by our previous results. More specifically, we consider a Bayesian Student linear regression similar to the Gaussian Bayesian regression in Section \ref{subsec:BayesianLinearRegression} but with a Student-distributed noise. The target posterior distribution is not conjugate with the Gaussian prior distribution. We can consider a Gaussian variational family ${\cal Q}$, but we are not in a recoverable case as the posterior distribution is not Gaussian and actually not of a known standard form.
Nevertheless, the Bonnet and Price estimator $g^N$ of the gradient (\ref{eq:BonnetPriceEstimator}) is tractable and can be computed in order to implement Algorithm~\ref{alg}.  

We assume 
$X \sim p_0 = \mathcal{N}(\mu_0,\Sigma_0)$ with $X \in \mathbb{R}^d$  and $Y \in \mathbb{R}$ with
\begin{equation*}
    \label{eq:StudentLikelihood}
    Y | X=x \sim \mathcal{S}(x^\top z, \sigma^2, \varrho),
\end{equation*}
where $z \in \mathbb{R}^d$ is a vector of fixed covariate values associated to $Y$, $\mathcal{S}(x^\top z, \sigma^2, \varrho)$ denotes the Student distribution with mean $x^\top z$, scale $\sigma^2$ and degrees-of-freedom parameter $\varrho$.  
Given $M$ data points $\{y_m\}_{m=1}^M$ and associated covariates $\{z_m\}_{m=1}^M$, the goal is to approximate the posterior $\pi (x | y_1, \ldots, y_M)$. 
In this setting, $g^N$ can be computed with (\ref{eq:BonnetPriceEstimator}) using that $\nabla \log \pi(x)$ and $\nabla^2 \log \pi(x)$ are available in closed-form,
\begin{align*}
    \nabla\log \pi(x) &= \sum_{m=1}^M \frac{(\varrho+1)\, (y_m - z_m^\top x)}{\varrho \sigma^2 + (y_m - z_m^\top x)^2}\, z_m - \Sigma_0^{-1} (x - \mu_0)\\
    \nabla^2 \log \pi(x) &= - \sum_{m=1}^M \left[ \frac{(\varrho+1)\, \big(\varrho \sigma^2 - (y_m - z_m^\top x)^2 \big)}{(\varrho \sigma^2 + (y_m - z_m^\top x)^2)^2}\right]z_m z_m^\top- \Sigma_0^{-1} \; .
\end{align*}

For a numerical illustration, we consider a subset of the \cite{dataset} dataset \citep{kaya2019}. We use the first $M=715$ samples given in this dataset for year 2013. The so-called turbine energy yield (TEY) variable is used as a response variable $Y$ and the  $d=10$ other  variables available in the dataset as fixed covariates. All variables are standardized for numerical stability.
To significantly depart from the Gaussian case, we set $\varrho=3$ to get an heavy-tailed noise distribution. We then set $\mu_0= 0$, $\Sigma_0 = 5 I$ and $\sigma^2=1$. 

The main difficulty in this setting is that $\pi$ is not log-concave, as it can be observed from the expression of $\nabla^2 \log \pi$ above. Therefore, there is a risk that the covariance adaptation fails, leading to singular covariance matrices, or covariance matrices very close to zero. In order to counteract this, we propose to constrain the search to the set $C = \{ (\mu, \Sigma + \mu \mu^{\top}) \in \mathcal{H} \text{ s.t. } \alpha I \preccurlyeq \Sigma \preccurlyeq \beta I\}$ defined in \eqref{eq:boundedEigenvalSet}. Here, we choose $\alpha = 10^{-4}$ and $\beta = 10^4$. The Bregman projection operator to this set is computed in Section \ref{subsec:projectionBoundedEigen}. Projection to this set ensures that we are working with covariance matrices whose conditioning is controlled. We run the algorithms initialized from the prior distribution, which is a Gaussian distribution  with mean $\mu_0$ and covariance $\Sigma_0$.

Results in terms of averaged ELBO are shown in Figure \ref{fig:StudentInfluenceProjection}, comparing the performance of Algorithm \ref{alg} with and without a projection step, in case of schedules with constant step sizes and constant or increasing sample sizes. We can observe that when no projection is performed, the ELBO stays constant across iterations, indicating a failure of the algorithm. However, adding the projection step allows to obtain a significant increase in ELBO, showcasing the positive impact of adding a projection step to counteract the lack of log-convexity of the target.

\begin{figure}[H]
    \centering
    \begin{subfigure}[t]{.49\linewidth}
        \centering
        \includegraphics[width=\linewidth]{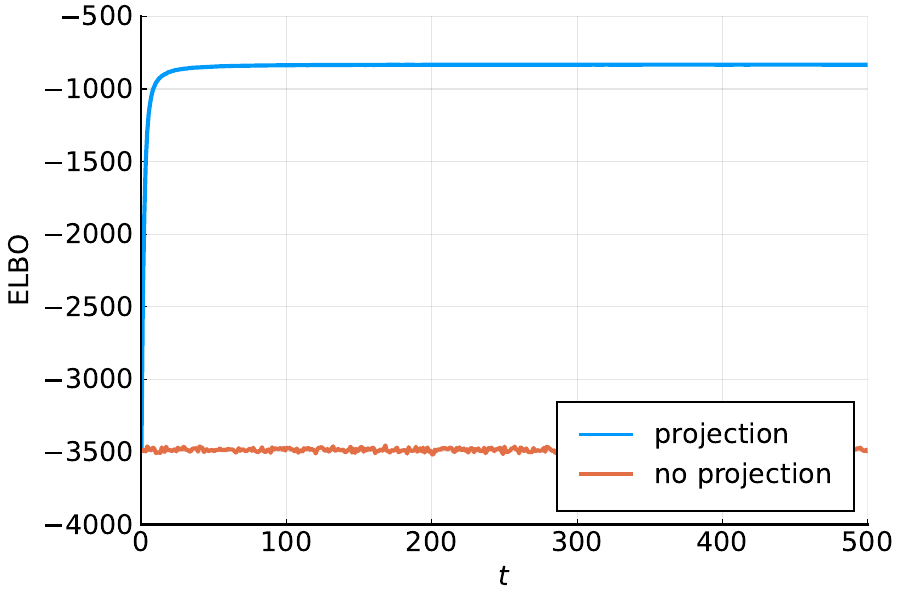}
        \caption{$\eta_t = 5\cdot 10^{-3}$ and $N_t = 250$}
        \label{fig:Student_constantEta_constantN}
    \end{subfigure}\hfill
    \begin{subfigure}[t]{.49\linewidth}
        \centering
        \includegraphics[width=\linewidth]{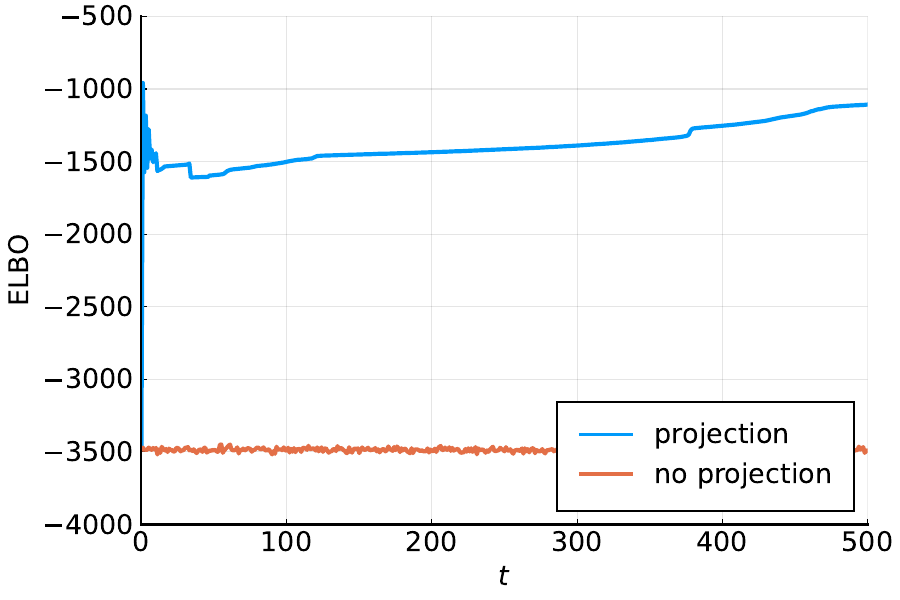}
        \caption{$\eta_t = 2/(t+2)$ and $N_t = 250$}
        \label{fig:Student_decreasingEta_constantN}
    \end{subfigure}

    \caption{Student linear regression: Averaged ELBO over $50$ runs, for different NGVI schedules, comparing for each schedule the algorithm with and without a projection step}
    \label{fig:StudentInfluenceProjection}
\end{figure}

Figure \ref{fig:Student_decreasingEta_constantN} shows that although the ELBO increases across iterations, this increase is non-monotonic compared to the one observed in Figure \ref{fig:Student_constantEta_constantN}. This behaviour can  be explained by the presence of outlier runs among the runs used to compute the averaged ELBO. Figure \ref{fig:StudentInfluenceProjectionConfidence} shows, for the same respective settings,  the median and inter-quartile intervals over $50$ runs of Algorithm \ref{alg} with projection. The observed monotonic convergence confirms that the non-monotonic behaviour observed in Figure \ref{fig:Student_decreasingEta_constantN} is due to the presence of a small number of outlier runs. If needed, these runs can be detected and stopped early. Note that compared to the schedule with constant step size, the schedule with decreasing step size imposes to take an initial step size equal to one, which amounts to forgetting about the initial condition, and may be the cause of this instability.

\begin{figure}[H]
    \centering
    \begin{subfigure}[t]{.49\linewidth}
        \centering
        \includegraphics[width=\linewidth]{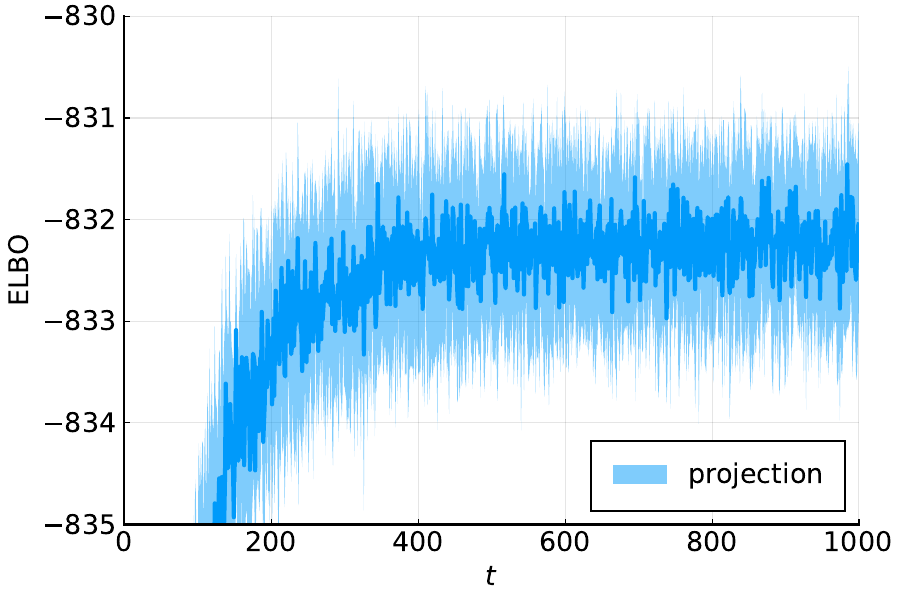}
        \caption{$\eta_t = 5\cdot 10^{-3}$ and $N_t = 250$}
        \label{fig:Student_constantEta_constantN_confidence}
    \end{subfigure}\hfill
    \begin{subfigure}[t]{.49\linewidth}
        \centering
        \includegraphics[width=\linewidth]{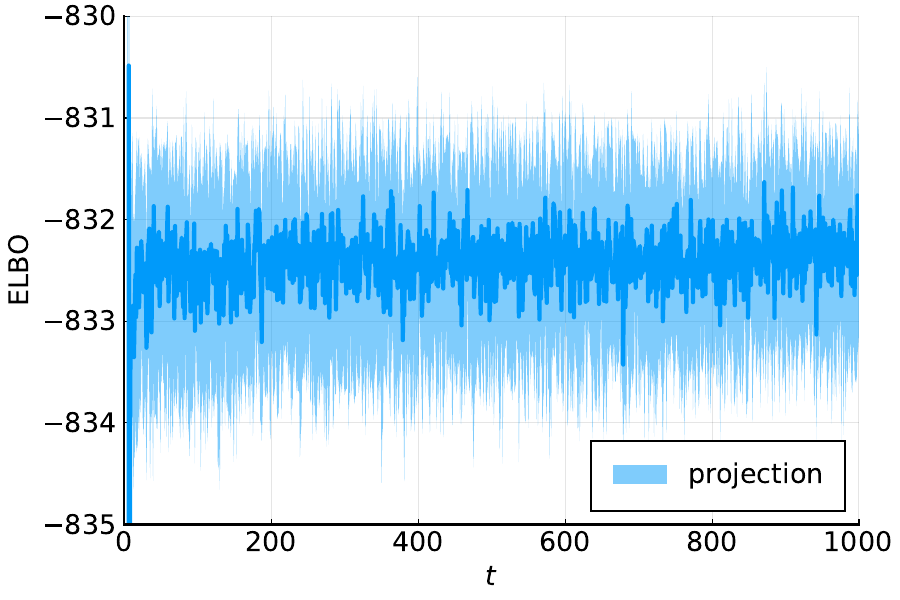}
        \caption{$\eta_t = 2/(t+2)$ and $N_t = 250$}
        \label{fig:Student_decreasingEta_constantN_confidence}
    \end{subfigure}

    \caption{Student linear regression: Median and intervals between the quantiles of order $0.25$ and $0.75$ of the ELBO over $50$ runs, for different NGVI schedules.}
    \label{fig:StudentInfluenceProjectionConfidence}
\end{figure}

We have not represented the performance obtained by taking sample size schedules of the form $N_t = (1+t)^\gamma$. In a fashion similar to the case of constant sample size and decreasing step size, increasing sample size schedules starting with $N_0 = 1$ may exhibit unstable behaviours. This can be avoided by choosing schedules of the form $N_t = \max( N, (1+t)^\gamma)$, which are covered by our theoretical results. Note also that we have not investigated the impact of  parameters $\alpha$ and $\beta$, which control the size of the set $C$. This would be an interesting perspective  to get more insights on the projection effect.

\subsection{Projection impact in a logistic regression setting}

We consider again the logistic regression task presented in Section \ref{sec:experiments}. Recall that we use Gaussian distributions with diagonal covariances to approximate the logistic regression posterior with a Gaussian prior. The LERC is not satisfied in this case. The experiment is conducted on synthetic data constructed with an actual regression vector $x_*$ whose components are all fixed being equal to $5$. Therefore, $x_* \in \mathbb{R}_{\geq 0}$ in this situation.

For each step size and sample size schedule, we investigate whether an additional projection step can improve the performance of Algorithm \ref{alg}. We consider the projection on the set $C = \{ (\mu, \sigma^2 + \mu \bullet \mu)\in \mathbb{R}^d \times \mathbb{R}^d \textrm{ s.t. } \mu \in \mathbb{R}^d_{\geq 0}\}$, whose associated projection operator is computed in Section \ref{subsec:projectionNonNegMean}. All the runs are initialized with initial mean simulated uniformly in $[-5,5]^d$ and initial covariance matrix being equal $0.5 I$. 

Results are depicted in Figure \ref{fig:logRegInfluenceProjection}, which shows that in this particular example, adding the projection operator improves upon the performance of the algorithm for every step size and sample size schedule. Indeed, we see that the ELBO obtained when using the projection is higher for the same iteration count than the ELBO obtained without the projection. Eventually, the two curves reach the same value.

\begin{figure}[H]
    \centering
    \begin{subfigure}[t]{.49\linewidth}
        \centering
        \includegraphics[width=\linewidth]{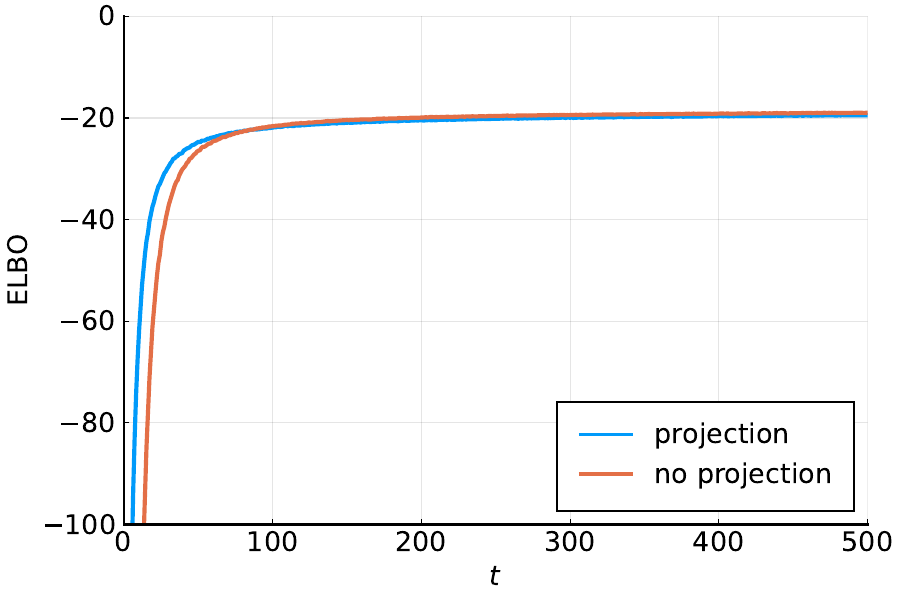}
        \caption{$\eta_t = 10^{-2}$ and $N_t = t+1$}
        \label{fig:logReg_constantEta_increasingN}
    \end{subfigure}\hfill
    \begin{subfigure}[t]{.49\linewidth}
        \centering
        \includegraphics[width=\linewidth]{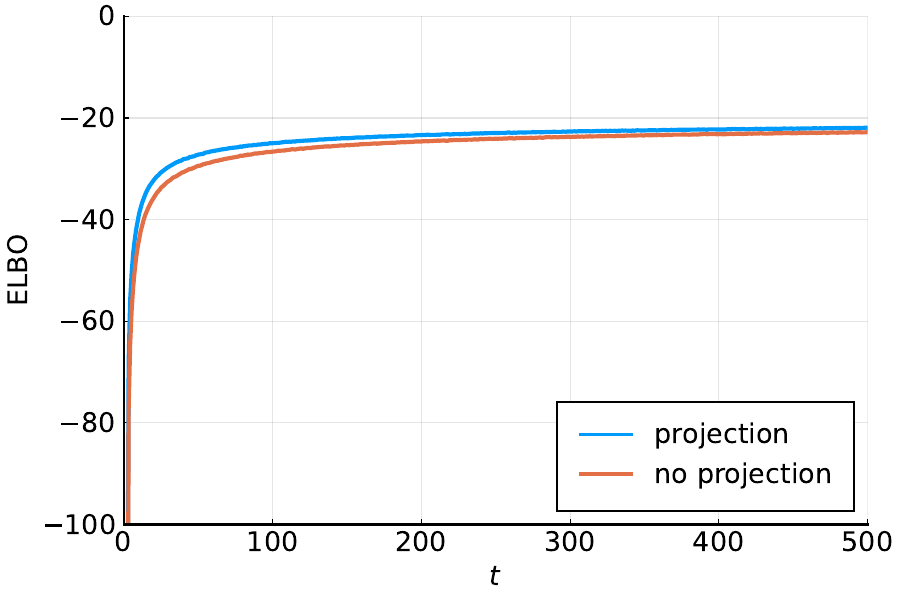}
        \caption{$\eta_t = 2/(t+2)$ and $N_t = t+1$}
        \label{fig:logReg_decreasingEta_increasingN}
    \end{subfigure}\hfill
    \begin{subfigure}[t]{.49\linewidth}
        \centering
        \includegraphics[width=\linewidth]{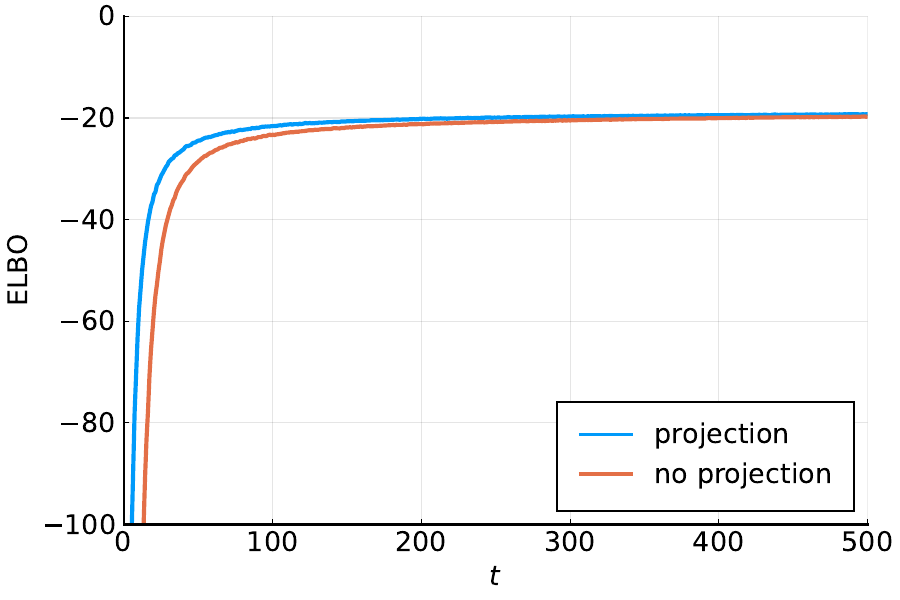}
        \caption{$\eta_t = 10^{-2}$ and $N_t = 10^2$}
        \label{fig:logReg_constantEta_constantN}
    \end{subfigure}\hfill
    \begin{subfigure}[t]{.49\linewidth}
        \centering
        \includegraphics[width=\linewidth]{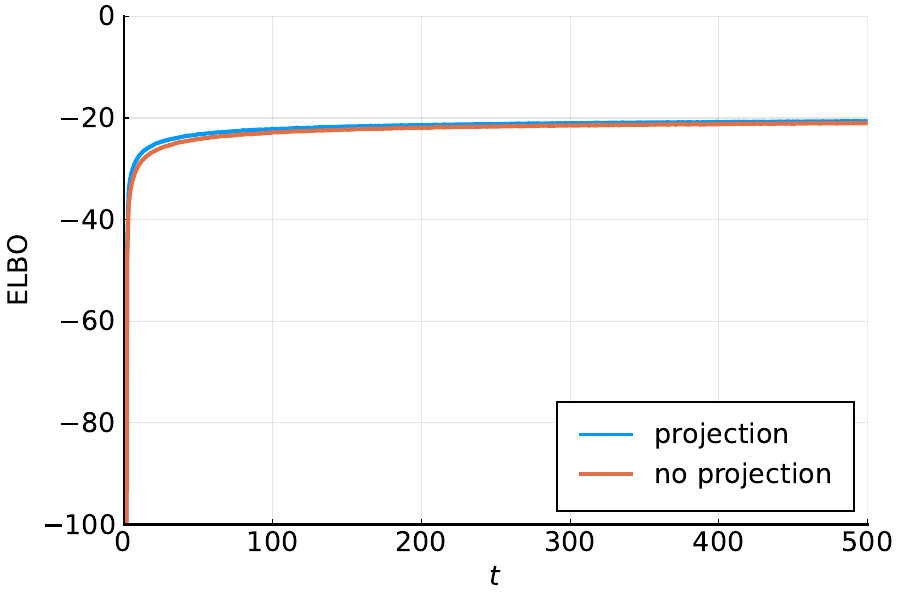}
        \caption{$\eta_t = 2/(t+2)$ and $N_t = 10^2$}
        \label{fig:logReg_decreasingEta_constantN}
    \end{subfigure}
    
    \caption{Logistic regression: Averaged ELBO over $50$ runs, for different NGVI schedules, comparing for each schedule the algorithm with and without a projection step}
    \label{fig:logRegInfluenceProjection}
\end{figure}

\end{document}